\documentclass[10pt,reqno]{amsart}
\usepackage{cite}
\usepackage[top=2cm,left=2.5cm,right=3cm,bottom=2cm]{geometry}
\newcommand{\rX}{\mathrm{X}}
\newcommand{\rL}{\mathrm{L}}
\newcommand{\cP}{\mathcal{P}}
\newcommand{\cD}{\mathcal{D}}
\newcommand{\cE}{\mathcal{E}}

\newcommand{\cU}{\mathcal{U}}
\newcommand{\hcU}{\hat{\cU}}

\newcommand{\cH}{\mathcal{H}}
\newcommand{\hcH}{\hat{\cH}}
\newcommand{\cL}{\mathcal{L}}
\newcommand{\cT}{\mathcal{T}}
\newcommand{\hcL}{{\hat{\cL}}}
\newcommand{\hcT}{{\hat{\cT}}}
\newcommand{\cW}{\mathcal{W}}

\newcommand{\lspan}{\operatorname{span}}

\newcommand{\ord}{\operatorname{ord}}

\newcommand{\hP}{{\hat{P}}}
\newcommand{\hH}{{\hat{H}}}
\newcommand{\hU}{{\hat{U}}}
\newcommand{\hW}{{\hat{W}}}

\newcommand{\hL}{\hat{L}}

\newcommand{\hphi}{{\hat{\phi}}}

\newcommand{\hT}{{\hat{T}}}

\newcommand{\hq}{{\hat{q}}}
\newcommand{\hr}{{\hat{r}}}
\newcommand{\hb}{{\hat{b}}}
\newcommand{\hw}{{\hat{w}}}

\newcommand{\Rset}{\mathbb{R}}

\newcommand{\Cset}{\mathbb{C}}
\newcommand{\Nset}{\mathbb{N}}

\newcommand{\al}{\alpha}

\newtheorem{thm}{Theorem}[section]

\newtheorem{conj}{Conjecture}[section]
\newtheorem{prop}{Proposition}[section]
\newtheorem{lem}{Lemma}[section]

\theoremstyle{definition}
\newtheorem{example}{Example}[section]
\newtheorem{definition}{Definition}[section]

\theoremstyle{remark}
\newtheorem{remark}{Remark}[section]

\begin{document}

\title[A conjecture on Exceptional Orthogonal Polynomials]{A conjecture on
  Exceptional Orthogonal Polynomials}

\author{David G\'omez-Ullate} \address{ Departamento de F\'isica
  Te\'orica II, Universidad Complutense de Madrid, 28040 Madrid,
  Spain} 

\author{ Niky Kamran } \address{Department of Mathematics and
  Statistics, McGill University Montreal, QC, H3A 2K6, Canada}

\author{Robert Milson} \address{Department of Mathematics and
  Statistics, Dalhousie University, Halifax, NS, B3H 3J5, Canada}

\begin{abstract}
  Exceptional orthogonal polynomial systems ($\rX$-OPS)  arise as
  eigenfunctions of Sturm-Liouville problems and generalize in this
  sense the classical families of Hermite, Laguerre and Jacobi. They also generalize the family of CPRS orthogonal polynomials introduced by Cari{\~n}ena {\it et al.}, \cite{CPRS}. We formulate the
following conjecture: \textit{every
    exceptional orthogonal polynomial system is related to a classical
    system by a Darboux-Crum transformation}.  We give a proof of
  this conjecture for codimension 2 exceptional orthogonal polynomials
  ($\rX_2$-OPs). As a by-product of this analysis, we prove a
  Bochner-type theorem classifying all possible $\rX_2$-OPS. The
  classification includes all cases known to date plus some new
  examples of $\rX_2$-Laguerre and $\rX_2$-Jacobi polynomials.
\end{abstract}
\maketitle
\section{Introduction}\label{sec:Intro}

The past several years have witnessed a considerable level of research
activity in the area of exceptional orthogonal polynomials, which are
new complete orthogonal polynomial systems arising as eigenfunctions
of Sturm-Liouville operators, extending the classical families of
Hermite, Laguerre and Jacobi.  The first examples of exceptional
orthogonal polynomial systems were discovered in \cite{GKM09a} and
\cite{GKM09b} as a result of the development of a direct approach
\cite{GKM3} to exact or quasi-exact solvability for spectral problems
in quantum mechanics that would go beyond the classical Lie algebraic
formulations \cite{Turbiner,KO,GKO}.

The exceptional orthogonal polynomial systems and the Sturm-Liouville problems
that define them have some key poperties that distinguish them from the classical orthogonal polynomial systems, and which we would
like to highlight. The most apparent one is that they admit {\it gaps} in their degrees, in
the sense that not all degrees are present in the sequence of
polynomials that form a complete orthonormal set of the underlying
weighted $\rL^{2}$ space, even though they are defined by a
Sturm-Liouville problem. This means in particular that they are not
covered by the hypotheses of Bochner's celebrated theorem on the
characterization of orthogonal polynomial systems defined by
Sturm-Liouville problems \cite{Bo}. 

The number of gaps in the sequence of degrees of the polynomials
appearing in a complete family will be referred to as the {\it
  codimension} and we will use the symbol $\rX_{m}$ to denote the
various complete orthogonal systems of codimension $m$.  A second order
differential operator is \emph{exceptional} if it preserves some
exceptional polynomial flag, but does not preserve the standard
polynomial flag generated by the monomials.  Thus, and in contrast
with the classical families where the defining differential operator
has only polynomial coefficients, the second order differential
operators corresponding to the exceptional families have poles in
their coefficients, although all their singular points happen to be
regular.

The first explicit examples of families of exceptional orthogonal
polynomials are the $\rX_1$-Jacobi and $\rX_1$-Laguerre polynomials,
which are of codimension one, and were first introduced in
\cite{GKM09a} and \cite{GKM09b}. In these papers, a characterization
theorem was proved for these orthogonal polynomial families, realizing
them as the unique complete codimension one families defined by a 
Sturm-Liouville problem. One of the key steps in the proof was the
determination of normal forms for the flags of univariate polynomials
of codimension one in the space of all such polynomials, and the
determination of the second-order linear differential operators which
preserve these flags \cite{GKM6,GKM12}.

It is Quesne \cite{Quesne1,Quesne2} who first observed the presence of a
relationship between exceptional orthogonal polynomials, the Darboux
transformation and shape invariant potentials. This enabled her to
obtain examples of potentials corresponding to orthogonal polynomial
families of codimension two, as well as explicit families of $\rX_2$
polynomials. Higher-codimensional families were first obtained by
Odake and Sasaki \cite{SO1}. The same authors further showed the
existence of two families of $\rX_m$-Laguerre and $\rX_m$-Jacobi
polynomials, the existence of which was explained in \cite{GKM10} for
$\rX_m$-Laguerre polynomials and in \cite{GKM12} for  $\rX_m$-Jacobi polynomials,
through the application of the isospectral algebraic Darboux
transformation first introduced in \cite{GKM04,GKM04a}. We also refer to
\cite{STZ10} for similar results, and to \cite{GKM10,GKM12} for the proof
of the completeness of the $\rX_m$-Laguerre families. We also note that some examples of exceptional Hermite polynomials were known to the
quantum physics community in the early 90s, \cite{dubov},
and are actively studied today under the name of CPRS systems,
\cite{CPRS,felsmith}. It should as well be noted that the exceptional Laguerre polynomials have already been used in a number of interesting physical contexts, for Dirac operators minimally coupled to external fields, \cite{Ho}, or in quantum information theory, \cite{Dutta-Roy1}.

The papers cited above contain many examples of orthogonal polynomial
families of arbitrary codimension arising from the Laguerre and Jacobi
system by the application of the Darboux transformation.  However, as
was shown in \cite{GKM12pub}, the above list is not exhaustive:
novel exceptional polynomials can be constructed by means of
\emph{multi-step} Darboux or Darboux-Crum transformations.  The multi-step idea was
further applied to exactly solvable and shape-invariant potentials up
in \cite{Grandati1,Quesne3,SO5}.  However, an essential question that
remains open is to know whether these families exhaust all the
possibilities of higher-codimensional complete orthogonal polynomial
systems, in other words whether {\it all} the higher-codimensional
complete orthogonal polynomial systems are generated by the
application of successive algebraic Darboux transformations. {\it We
  conjecture this result to be true.} In order to prove such a result,
one approach would be to try to carry out for all codimensions an
analysis similar to the one performed in \cite{GKM6, GKM09a, GKM09b}
in codimension one, identify the complete orthogonal sets amongst the
resulting families and show that all of these can be obtained from the
classical codimension-zero families by iterating algebraic Darboux
transformations (we will refer to these as {\em multi-step} Darboux
transformations). This seems like a very challenging task in the
absence of a general classification strategy that would lead to normal
forms for flags of univariate polynomials for all codimensions. Even
in the codimension two case, the question would be quite difficult to
answer if we were only using the tools that were at our disposal in
\cite{GKM6}.  We are nevertheless able to give a complete answer to
this question for codimension two families by suitably refining the
approach taken in these earlier papers. In particular the possible
pole structure of the coefficients of the operators that preserve the
codimension two flags plays a key role in the analysis.

Since the main objects of our study are orthogonal polynomial systems that
arise as eigenfunctions of a Sturm-Liouville problem, let us give a definition:

\begin{definition}
  \label{def:OPS}
  We define a Sturm-Liouville orthogonal polynomial system (SL-OPS) as a
sequence of real polynomials $y_1(x), y_2(x), y_3(x),\ldots$, with $\deg y_i>\deg
y_j$ if $i>j$, satisfying the following conditions:
\begin{itemize}
\item[(i)] There exists a second order differential operator 
\[T[y]=p(z) y''+ q(z) y' + r(z) y\] 
with rational coefficients $p,q,r$ such that
$T[y_i]=\lambda_i y_i$ for all $i$, with the $\lambda_i$ distinct.

\item[(ii)] There exists an interval  $I=(a,b),\; -\infty\leq a<b\leq \infty$
such that the weight function 
\[ W(x) = \frac{1}{p(x)} \exp\left(\int^x \frac{q}{p} dx\right) \]
is positive, that is $W(x) > 0$ for $x\in I$, such that all moments are
finite,
\[ \int_a^b x^i W(x) dx <\infty,\quad i=0,1,2,3\ldots, \] 
and such that $p(x)W(x)\to 0$ at the endpoints $x=a,b$.

\item[(iii)] The polynomial sequence is complete, meaning that $\lspan \{ y_i
    \}_{i=1}^\infty$ is dense in $\rL^2(Wdx, I)$;ß

  \end{itemize}

\end{definition}
\noindent
\begin{remark}
It follows from the above definition that the operator $T$ is essentially self-adjoint on the weighted Hilbert space
$L^2(I,Wdx)$ and that the eigenpolynomials are orthogonal, meaning that
\[ \int_a^b W y_i y_j dx = k_i
    \delta_{ij} ,\quad k_i>0,\]
for some constants $k_i$.
\end{remark}
\begin{remark}
If $\deg y_i = i-1$ for all $i$, we are dealing with one of the
\textit{classical} orthogonal polynomial systems of Hermite, Laguerre and Jacobi: the polynomials in
question span
the standard polynomial flag and $p,q,r$ are polynomials of degrees $2,1$ and
$0$ respectively, \cite{Bo}.
\end{remark}
\begin{remark}
If the degree sequence $\{\deg y_i\}_{i=1}^\infty$  does not contain all
non-negative integers, then we will have an exceptional polynomial system
($\rX$-OPS), and the coefficients of $T$ will be purely rational
(as opposed to polynomial) functions.
\end{remark}
\begin{remark} We shall see in Section 5.1 that the eigenvalue equation $T[y]=\lambda y$ can be put into Sturm-Liouville form.
\end{remark}

Even though several families of $\rX$-OPS have now been described in the literature, the general question of classifying all such systems is still largely open. In particular, major porgress would be achieved in our understanding of the subject if we could obtain a classification or a characterization of all families of SL-OPS. (Recall the classification performed
by Bochner \cite{Bo} and Lesky \cite{Lesky} deals only with the classical OPS.) It seems clear
by now that the Darboux transformation will be one of the key necessary tools in achieving such a goal.
It should be noted that when referring to the Darboux transformation, we do not mean here the factorization of Jacobi matrices into upper triangular and lower triangular matrices \cite{GH}. Indeed, while such a transformation is defined for any OPS, the transformed OPS will in
general not be a SL-OPS even if the original OPS was one.
We will rather use \textit{algebraic Darboux transformations}\footnote{A wider class
of these transformations has been extensively used in Quantum Mechanics to
generate new exactly solvable problems from known ones. The subclass of interest
to us in the context of OPS consists of the set of transformations that preserve the polynomial character of the eigenfunctions. This particular class of
Darboux transformations was characterized in \cite{GKM04,GKM04a}.}, also known as rational factorizations, which are defined only for SL-OPS. In these
transformations, it is the second order operator $T$ that needs to be factorized
as the product of two first order operators $T=AB$, and the transformed operator $\hat T$ is obtained
by reversing the order of the factors, $\hat T=BA$. We shall see that by construction, these algebraic
Darboux transformations transform an SL-OPS into another SL-OPS. 

We are now ready to state the main result of our paper:
\begin{thm}\label{thm:Main}
Every $\rX_m$ orthogonal polynomial system for $m\leq 2$ can be obtained
by applying a sequence of at most $m$ Darboux transformations to a classical
orthogonal polynomial system. 
\end{thm}
The proof of this theorem is done in several steps. The first step,
carried out in Section \ref{sect:Classif}, consists in the
classification of $\rX_2$ flags and the
determination of the corresponding pole structure for the coefficients of
the second order linear differential operators that preserve
them. This forms the substance of Theorem \ref{thm:X2flag}.  It should
be noted that in contrast with the codimension one case, the canonical
codimension two flags contain free parameters (flag moduli).
In Section 5 we provide the necessary background to select from the
classification of $\rX_2$-flags those that give rise to a well defined SL-OPS.
This selection is performed in Section 6, where
Theorem \ref{thm:x2op} provides the classification of  $\rX_2$ orthogonal
polynomial systems. It is worth noting that this classification contains new
families of $\rX_2$-Laguerre and $\rX_2$-Jacobi polynomials; for example
the new Laguerre-type family of Section 6.5.6 with weight $e^{-x}
x^{1/4} / (4x+3)^4$. The second step in
the proof of Theorem \ref{thm:Main} , which is carried out in Section
\ref{sect:fact}, consists
of the proof of the key property, stated in Theorem \ref{thm:ExDarb},
that all $\rX_1$ and $\rX_2$ operators are related to a classical
operator by a Darboux transformation or a
sequence of two Darboux transformations.

Finally, we will conclude by stating our general, yet-to-be proved,
conjecture, which extentends the result of Theorem \ref{thm:Main} to arbitrary
codimension.
\begin{conj}
Every $\rX_m$ orthogonal polynomial system for any codimension $m$ can be
obtained by applying a sequence of at most $m$ Darboux transformations to a
classical OPS. 
\end{conj}

\section{Preliminaries and definitions}

\subsection{Polynomial flags}
Let $\cP$ denote the infinite-dimensional space of real, univariate
polynomials, and let $\cP_n\subset \cP$ be the $n+1$ dimensional
subspace of polynomials having degree $n$ or less.  We define the
degree of a finite dimensional polynomial subspace
$U\subset \cP$ to be
\begin{equation}
  \label{eq:degcodimdef}
  \deg U = \max \{ \deg p: p \in \cU \}.
\end{equation}
\begin{definition}
  A \emph{polynomial flag} is an infinite sequence of polynomial
  subspaces $U_1\subset U_2\subset \ldots$, nested by inclusion, such
  that $\dim U_k = k$, and such that $\deg U_k < \deg U_{k+1}$ for all
  $k$.  The \emph{total space} of a polynomial flag is the
  infinite-dimensional polynomial subspace
  \begin{equation}
    \label{eq:cupUk}
    \cU = \bigcup_{k=1}^\infty  U_k.
  \end{equation}
\end{definition}

\begin{definition}
  Let $\cU \subset \cP$ be an infinite dimensional polynomial
  subspace.  A degree-regular basis of $\cU$ is a sequence of
  polynomials $\{p_k\}_{k=1}^\infty$ such that $ \deg p_k < \deg
  p_{k+1}$ and such that $ \cU = \lspan \{ p_k \}$.
\end{definition}
\begin{prop}
  Let $U_1\subset U_2 \subset \cdots$ be a polynomial flag, $\cU$ the
  total space, and $\{ p_k\}_{k=1}^\infty$ a degree regular basis.
  Then, for all $k=1,2,\ldots$, we have
  \[ U_k = \lspan \{ p_1, \ldots, p_k \}.\]
\end{prop}
\begin{prop}
  Let $\cU\subset \cP$ be an infinite dimensional polynomial subspace.
  Let $\hat{U}_k\subset \cU$ be the unique $k$-dimensional subspace
  having minimal degree. Then $\hat{U}_1 \subset \hat{U}_2 \subset
  \cdots$ is a polynomial flag whose total space is $\cU$.
\end{prop}
\begin{prop}
  Let $U_1\subset U_2\subset \ldots$ be a polynomial flag and $\cU$
  the corresponding total space.  Let $\hat{U}_k$ be as above.  Then,
  $\hat{U}_k = U_k$.
\end{prop}
\noindent
The above propositions show that there is a natural bijection between the set of
polynomial flags
and the set of infinite-dimensional polynomial subspaces.  Going
forward it will sometimes be more conveninient to specify the total
space rather than the actual flag.  The identification of the flag and
its total space will be implicitly assumed. We will use the complete notation
$\cU:U_1\subset U_2\subset\dots$ for the flag and its total space, but we will
write only $\cU$ to denote the flag $U_1\subset U_2\subset\dots$ where no
confusion can arise.

\begin{definition}
 Given a polynomial flag $\cU:U_1\subset U_2\subset\dots$ , we define the
\textit{degree sequence} $\{n_k\}_{k=1}^\infty$ and the \textit{codimension
sequence} $\{m_k\}_{k=1}^\infty$ as
\begin{equation}
  \label{eq:mkdef}
  n_k = \deg U_k,\quad   m_k= n_k+1-k.
\end{equation}
where $m_k$ is the codimension of $U_k$ in $\cP_{n_k}$. 
\end{definition}
It is easy to see that
$\{n_k\}$ is strictly increasing while $\{m_k\}$ is
non-decreasing. In this paper we will focus on flags with \emph{finite
codimension}, which means that the total space $\cU$ has finite codimension in
$\cP$, or equivalently, that the codimension sequence $\{m_k\}$ admits
an upper bound $m=\max_k m_k$, which we call the \textit{codimension of the
flag}. As mentioned in the Introduction, one might also characterize $m$ as the
number of gaps in the degree sequence.  We will say that a polynomial flag
has \emph{stable codimension} if $m_k=m$ for all
$k$, or equivalently if the degree sequence satisfies $n_1=m$ and
$n_{k+1}=n_{k}+1$ for all $k\geq 1$.

The simplest of all polynomial flags in the \textit{standard flag}
$\cU_{\text{st}}:\cP_0\subset\cP_1\subset\cP_2\subset\dots$. The total space
for this flag is $\cP$, its degree sequence is $\mathbb N \cup \{0\}$ and it
has stable codimension zero.

\begin{definition}
  We will say that  a second order differential operator 
  \begin{equation} \label{eq:Tydef} T[y(z)] = p(z) y'' + q(z) y' + r(z) y,
  \end{equation}
  is \emph{rational}, if the coefficients $p,q,r$ are rational
  functions of the independent variable $z$ and the prime denotes derivation
with respect to this variable, $y'=\frac{dy}{dz}$. The poles of a rational
  operator $T$ are the poles of $p,q$ and $r$. An operator $T$ with no
  poles is said to be polynomial.  If there is one or more poles, then
  we will refer to  $T$ as \emph{non-polynomial}.
\end{definition}

\begin{definition}
  We say that a polynomial flag $\cU:U_1\subset U_2\subset\dots$ is invariant
under a rational operator $T[y]$ if $T(U_k) \subset U_k$ for all $k$.  We
let $\cD_2(\cU)$ denote the vector space of all second order operators that
  preserve the flag $\cU$.\footnote{We stress that invariance of the whole flag
$\cU:U_1\subset U_2\subset\dots$ is a much stronger condition than the invariance
of the total space $\cU$. For the purpose of this study, we will always require
invariance of the flag, since this condition leads to polynomial
eigenfunctions of the operator.}
\end{definition}

In the analysis of invariant polynomial flags, no generality is lost
by considering only second order operators with rational coefficients, as
evidenced by the following
\begin{prop}
  \label{prop:cramer}
  Let $T[y] = py''+qy'+ry$ be a differential operator such that
  \[T[y_i] = g_i,\; i=1,2,3,\]
  where $y_i, g_i$ are polynomials with
  $y_1, y_2, y_3$ linearly independent.  Then, $p,q,r$ are rational
  functions. 
\end{prop}
\begin{proof}
  It suffices to apply Cramer's rule to solve the linear system
  \[ \begin{pmatrix}
    y''_1 & y'_1 & y_1\\
    y''_2 & y'_2 & y_2\\
    y''_3 & y'_3 & y_3
  \end{pmatrix}
  \begin{pmatrix}
    p\\q\\r
  \end{pmatrix}
  =
  \begin{pmatrix}
    g_1\\ g_2\\ g_3
  \end{pmatrix}
  \]
\end{proof}

\begin{definition}
  A polynomial flag is \emph{imprimitive} if it admits a non-trivial
  common factor.  Otherwise, the flag is said to be \emph{primitive}.
\end{definition}

\begin{prop}
  \label{prop:gaugeequiv}
  Let $\cU$ be a primitive flag, let $\mu$ be a polynomial of degree
  $\geq 1$ and let
  \[ \tilde{\cU} = \mu \cU = \{ \mu p : p \in \cU\}.\] be the
  corresponding imprimitive flag.  Suppose that $T[y]$ is a rational
  operator that preserves $\cU$. Then, the gauge-equivalent rational
  operator $\tilde{T} = \mu T \mu^{-1}$ preserves $\tilde{\cU}$.
\end{prop}
\noindent
Therefore, primitive flags can be regarded as canonical
representatives for the equivalence relation modulo gauge
transformations, and we can restrict our attention to primitive flags in the
classification of invariant polynomial flags. The main object of our study is
then the class defined in the following definition.
\begin{definition}
  A second order operator that preserves a primitive polynomial flag, but
  does not preserve the standard flag will be called \emph{an
    exceptional operator}.  An \textit{exceptional flag} is the
  \emph{maximal} primitive polynomial flag that is preserved by a second
  order exceptional operator.  Exceptional flags and operators of
  finite codimension $m\geq 1$ will henceforth be called $\rX_m$ flags
  and operators.  By contrast, a second order
  differential operator that preserves the standard flag
  $\cP$, will be referred to as a \emph{classical operator}.
\end{definition}

\begin{thm}[Bochner]
  \label{thm:bochner}
  A classical operator has the form
  \[ T[y] = p y'' + q y' + r y \]
  where $p\in \cP_2,\; q\in \cP_1$ are polynomials of the indicated
  degree, and where $r$ is a constant.
\end{thm}
\begin{prop}
  \label{prop:polyop}
  An exceptional operator is, necessarily, non-polynomial.
\end{prop}
\begin{proof} See the proof of Lemma 3.1 in \cite{GKM12}.
\end{proof}
\noindent
Thus, an exceptional operator has poles, but it also has an infinite number
of polynomial eigenfunctions.  When classifying exceptional flags by
increasing codimension, each flag will give rise to new operators not
considered at lower codimension, which justifies the definition above.
Here are some examples to illustrate these definitions
\begin{example}
  \label{ex:x1unstable}
  The flag with basis $\{ 1,z^2,z^3,\ldots,\}$ is exceptional because
  the operator
  \[T[y] = y''-\frac{2y'}{z}\] preserves the flag. The degree sequence is
  $\{0,2,3,\dots\}$ and the codimension sequence is $\{0,1,1,\dots  \}$ so the
  flag has non-stable codimension 1.
\end{example}

\begin{example}
  \label{ex:x1stable}
  By contrast, the flag spanned by $\{z+1, z^2, z^3,\ldots\}$ has a
  stable codimension $m=1$.  This flag is exceptional because it is
  preserved by the operator
  \[ T[y] =y''-2\left(1+\frac{1}{z}\right)y'+\left(\frac{2}{z}\right) y.\]
\end{example}
\begin{example}
  Let $H_k(z)$ denote the degree $k$ Hermite polynomial.  The
  codimension 1 flag spanned by $\{H_1, H_2, H_3,\ldots\}$ is not
  exceptional. The flag is preserved by the operator $T[y] = y''-z
  y'$. However, this operator also preserves the standard flag, which violates
the maximality assumption.
\end{example}
\begin{example}
  The codimension 1 flag spanned by $z,z^2,z^3, \ldots$ is
  preserved by the operator
  \[ \tilde{T}[y] = y'' - \frac{2y'}{z}+ \frac{2y}{z^2}.\]  This is not an
  exceptional flag because it is imprimitive as $z$ is a non-trivial
common factor. In fact, the operator $\tilde T$ is gauge equivalent 
$\tilde{T} = z T z^{-1}$ to the operator $T[y] = y''$  that
  preserves the standard flag.
\end{example}

\begin{example} 
  \label{ex:X2unstable}
  Let
  \begin{equation}
    \label{eq:codim2example}
    \left\{\begin{array}{l} y_{2k-1} = z^{2k-1} - (2k-1)z,\\
        y_{2k} = z^{2k} - k z^2,
      \end{array}\right.
    \qquad k=2,3,4,\ldots
  \end{equation}
  Consider the flag spanned by $\{1,y_3, y_4, y_5,\ldots\}$. The degree sequence
  of the flag is $0,3,4,5,\ldots$ so
  it is a non-stable codimension 2 flag. The flag is preserved by the following
  operators :
  \begin{align}
    T_3[y] &= (z^2-1)y''-2z y',\\
    T_2[y] &= zy''-2\left(1+\frac{2}{z^2-1}\right)y',\\
    T_1[y] &= y'' +z\left(1-\frac{4}{z^2-1}\right)y'.
  \end{align}
  The flag is exceptional, because $T_1$ and $T_2$ do not preserve the
  standard flag.  Since $T_2, T_1$ have 2 distinct
  poles, they do not preserve a codimension 1 flag (see Lemma
  \ref{lem:1stordcond}).
\end{example}

\section{Classification of exceptional codimension 2 polynomial
  flags}\label{sect:Classif}
In this Section we perform a classification of all $\rX_2$-flags up to
affine transformations of the independent variable $z$. We exhibit
degree-regular bases for each of them, and we determine the
$\rX_2$-operators that preserve them.  We begin by introducing the
following flags: {\small
  \begin{align}
    \label{eq:EX1def}
    \cE^{(1)}(a;b) &:= \{ p\in \cP : p'(b) = a p(b) \} \\
    \label{eq:E1def}
    \cE^{(11)}(a_0,a_1;b_0,b_1) &:= \cE^{(1)}(a_0;b_0)\cap \cE^{(1)}(a_1;b_1) \\
    \cE^{(2)}(a_{01}, a_{03}, a_{23};b) &:= \{ p \in \cP : p'(b) =
    a_{01}
    p(b),\; p'''(b) = 3 a_{23} p''(b) + 6a_{03} p(b) \} \label{eq:E2def}
  \end{align}} 
\noindent
The first flag has codimension one and its associated
$\rX_1$-operator will have a simple pole at $z=b$. The second flag has
codimension two and its associated $\rX_2$ operator will have two
simple poles at $b_0$ and $b_1$. The third flag has codimension two
and its associated $\rX_2$ operator will have a simple pole at $b$. The
notation in the superindices is connected to the the order of the
poles of the weight for the exceptional orthogonal polynomial system
based on the flag. This will become clear in Section \ref{sec:XOPs}.
In any case, the sum of superindices must always coincide with the codimension
of the flag.

Some, but not all of the parameters in the above flags can be
normalized by means of an affine transformation.  Thus, unlike the
codimension one case, the $\rX_2$ flags contain free continuous
parameters, which shall be refereed to as \textit{flag moduli}. 
As explained before, the parameters $b, b_0$ and $b_1$ will be the positions of
the poles of the operators. If there is one pole we will set $b=0$ and if
there are two poles we will normalize them as $b_0= 0,\; b_1 = 1$.
Note that any two poles in the complex plane can be transformed into $0$ and
$1$ by a complex affine transformation, so there is no loss of generality
involved in the above normalization.

Below, we describe each of the above flags in terms of a basis. {\small
  \begin{align}
    \label{eq:E1span}
    \cE^{(1)}(a;0) &= \lspan \{ 1+az, z^2, z^3,  z^4,\ldots \} \\
    \label{eq:E11span}
    \cE^{(11)}(a_0,a_1;0,1) &= \lspan\{ z^2((a_1-2)(z-1)+1),
    (z-1)^2((a_0+2)z+1)\}\cup\\ \nonumber
    &\qquad \qquad\qquad\{ z^2(z-1)^2 z^j\}_{j=0}^\infty,\\
    \label{eq:E2span}
    \cE^{(2)}(a_{01},a_{03},a_{23};0) &= \lspan \{1+a_{01} z + a_{03}
    z^3 , z^2+ a_{23} z^3 
    ,z^4,z^5,\ldots\}
  \end{align}} 

Let us first recall the main result of the classification of $\rX_1$-flags
first proved in \cite{GKM09a} (see \cite{GKM12} for a more recent and
streamlined proof).
\begin{thm}
  \label{thm:X1flag}
  Every stable $\rX_1$ polynomial flag is affine-equivalent to
  \[\cE^{(1)}(1;0) = \lspan \{ 1+z,z^2,z^3,z^4,\ldots \}.\]
  Every unstable $\rX_1$ polynomial flag is
  affine-equivalent to the monomial flag 
  \[ \cE^{(1)}(0;0) = \lspan \{ 1,z^2,z^3,z^4,\ldots \} .\]
\end{thm}
\noindent
Note that, as mentioned before, the most general $\rX_1$ flag up to
affine transformations contains no flag moduli. The main result of
this section is the following theorem that describes the situation for
$\rX_2$ flags.
\begin{thm}
  \label{thm:X2flag}
  Up to an affine transformation every $\rX_2$ flag is equivalent to one
  of the following two flags:
\begin{enumerate}
 \item  $\cE^{(11)}(a_0,a_1;0,1)$ 
\item   $\cE^{(2)}(a_{01},a_{03}, a_{23};0)$  subject to the constraint
  \begin{equation}
    \label{eq:E2constraint}
    a_{03} (a_{01} - a_{23}) (6 a_{03} + a_{01} a_{23} (a_{01} +
    a_{23})) = 0
  \end{equation}
\end{enumerate}
\end{thm}

Before we can address the proof of the above theorem, we need to introduce 
more concepts and establish some key intermediate results.

For a polynomial $y(z)$ and a constant $b\in \Cset$, we define $\ord_b
y\geq 0$ to be the order of $b$ as a zero of $y(z)$.  Let $\cU\subset
\cP$ be a polynomial subspace. For $b\in \Cset$ define
\begin{equation}
  \label{eq:IbUdef}
  I_b(\cU)  = \{ \ord_b y : y\in \cU \}.
\end{equation}
\begin{lem}
  \label{lem:ordergaps}
  Let $T$ be a rational operator that preserves a primitive
  polynomial subspace   $\cU\subset \cP$.
  Let
  \[ T = \sum_{i=-d}^\infty T_i ,\]
  where
  \[T_i[y] = z^i \left(p_i z^2 y'' + q_i z y' + r_i y\right)\] for some
  constants $p_i, q_i, r_i$ be the degree-homogeneous representation
  of $T$ in terms of Laurent series.  If $T$ has a pole at $z=0$,
  then $d=2$, $r_{-2} = 0$, and there exists a positive integer
  $\alpha \geq 1$ such that
  \begin{equation}
    \label{eq:Ibform}
    I_0(\cU) =   \Nset / \{ 1,3,\ldots, 2\alpha-1 \}
  \end{equation}
\end{lem}
\begin{proof}
  Observe that $T_i$ is  degree-homogeneous, meaning that
  \[ T_i[z^j] = (p_i j(j-1) + q_i j + r_i)z^{i+j}.\] So either $T_i$ 
  annihilates a given monomial $z^j$, or it shifts its degree by $i$.
  A non-zero $T_i$ can annihilate at most two distinct
  monomials, whose exponents $j$ satisfy the quadratic constraint
  \[ p_i j(j-1) + q_i j + r_i = 0 .\] By definition, $i\in I_0$ if and
  only if the flag contains a polynomial of the form $z^i +$ higher
  degree terms.  Since $T$ preserves $\cU$ and since $T_{-d}$ is the
  leading term of the operator, it follows that $T_{-d}$ preserves the
  monomial subspace $\{ z^i : i \in I_0\}$.
  
  For $T$ to have a pole at $z=0$ we must have $d>0$. Since $\cU$ is
  primitive, $0\in I_0$ and therefore $T_{-d}$ must annihilate $z^0=1$. Observe
that the leading order $d$ must also be $d\geq 2$, since $d=1$ would
require that $T_{-1}[1]=0 \Rightarrow r_{-1} = 0$, so operator $T$ would be
polynomial, contrary to the hypothesis. To conclude the proof, we will establish
that $d$ has to be precisely $2$.  Since the flag $\cU$ has finite
codimension, there are only a finite number of gaps (missing integers) in the
set $I_0$. Let $i\notin I_0$ be one such gap, then
  either $i+d\notin I_0$, or $T_{-d}$ annihilates $z^{i+d}$.  Hence,
  $1\notin I_0$ must be a gap.  Otherwise, since $d\geq 2$, $T_{-d}$ would need
to annihilate three monomials: $z^0,z^1$ and at least one higher degree
monomial, which is impossible.
  Thus, for some integer $\alpha \geq 1$, the gaps in the $I_0$ sequence are
  $1,1+d,1+2d,\ldots, 1+d(\alpha-1)\notin I_0$, with $T_{-d}[z^{d\alpha+1}]=0$. 
Note that $T_{-d}$ annihilates $1$ and  $z^{d\alpha+1}$ so it cannot annihilate
any other monomial and therefore the above gaps are the only possible gaps in
$I_0$.
It follows that $2\in I_0$ is \emph{not} a gap. If the leading order was $d>2$
  then $T_{-d}$ would be required to annihilate also $z^2$, which is impossible.
We conclude then that $d=2$ and since $T_{-2}[1]=0$  we must
have $r_{-2} = 0$. The assertions of the lemma are proved.
\end{proof}

The following lemma shows how to decompose a rational second order operator
that preserves a primitive polynomial flag.

\begin{lem}
  \label{lem:Tansatz}
  Let $T$ be a second order rational operator with poles $b_1, \ldots,
  b_N\in \Cset$.  If $T$ preserves a primitive polynomial flag of
  finite codimension, then necessarily it has the form
  \[ T[y]= p_{-2} y'' + (p_{-1} z y'' + q_{-1} y') + (p_0 z^2 y''+ q_0
  y' + r_0 y) + \sum_{i=1}^N c_i \,\frac{y'- a_i y}{z-b_i},\] where 
  $p_i, q_i, r_i\in\mathbb R$ and $a_i, c_i\in\mathbb C$ are constants.
\end{lem}
We see therefore that an exceptional operator must have rational coefficients
that can only contain simple poles.

\begin{proof}
  We decompose the given operator as
  \[ T = \sum_{i=0}^N T^{(i)} \] where $T^{(0)}$ is a polynomial
  operator and where
  \[ T^{(i)}[y] = \frac{r^{(i)}_{-1} y}{z-b_i} + \frac{q^{(i)}_{-2}
    (z-b_i)y' + r^{(i)}_{-2}y}{(z-b_i)^2} + \sum_{j=3}^{d_i}
  \frac{p_{ij} (z-b_i)^2 y'' + q_{ij} (z-b_i) y' + r_{ij}
    y}{(z-b_i)^j} \] for some positive integer $d_i\geq 1$ and
  constants $p_{ij}, q_{ij}, r_{ij}$.

  Let $\cU$ be the total space of the primitive flag preserved by $T$.
  Since $T(\cU)\subset \cP$, it follows that $T^{(i)}(\cU) \subset
  \cP$ for every $i=0,1,\ldots, N$.  By construction, the operators
  $T^{(1)}, \ldots, T^{(N)}$ all lower degrees.  Since $T$ preserves
  an infinite flag, it cannot have a degree raising part.  Therefore,
  $T^{(0)}$ has the form
  \[ T^{(0)}[y]= p_{-2} y'' + (p_{-1} z y'' + q_{-1} y') + (p_0 z^2
  y''+ q_0 y' + r_0 y) \]
  
  Expanding the operator coefficients as Laurent series in $z-b_i$, we
  apply Lemma \ref{lem:ordergaps} to conclude that $d_i=2,
  r^{(i)}_{-2} = 0$ for all $i=1\ldots, N$.  The desired conclusion
  has been established.
\end{proof}
\noindent
Note that if $b_i$ is a real pole, then the constants $a_i$ and $c_i$ must
also be real since the flag is real too.
The next lemma shows that for every pole $b_i$ of an exceptional differential
operator, the elements of its invariant flag must satisfy a first order
differential constraint at that pole.
\begin{lem}
  \label{lem:1stordcond}
  Let $T[y]$ be a second order rational operator with poles $b_1,
  \ldots, b_N$ that preserves a primitive flag $\cU$ of finite codimension.
  Then, there exist constants $a_1,\ldots, a_N$ such that the elements
  of $y\in\cU$ obey 1st order differential constraints of the form
  \[ y'(b_i) = a_i y(b_i),\quad i=1,\ldots, N. \]
\end{lem}
\begin{proof}
  By Lemma \ref{lem:ordergaps}, for each $i=1,\ldots, N$ the total
  space $\cU$ contains a polynomial of the form
  \[ y^{(i)}_0(z) = 1 + a_i (z-b_i) + O\big((z-b_i)^2\big) ,\]
  but does not contain an element of the form
  \[ (z-b_i) + O\big((z-b_i)^2\big).\]
  Therefore, every $y\in\cU$ either starts as $y^{(i)}_0(z)$ or at degree $2$
in $(z-b_i)$ so in any case it obeys the constraint $y'(b_i) = a_i
  y(b_i)$.
\end{proof}

At this point, it becomes necessary to describe and analyze certain
degenerate subclasses of the $\cE^{(11)}$ and $\cE^{(2)}$ flags defined in
\eqref{eq:E1def}-\eqref{eq:E2def}.
The distinguishing property of these subclasses is the first two elements
of the degree sequence of the  flag.  Thus, when we write $\cE_{ij}$
where $0\leq i<j\leq 3$ we are referring to a codimension two flag whose degree
sequence is $\{i,j,4,5,6,\dots\}$. 
The generic case is the stable codimension two flag $\cE_{23}$, which starts
at degree $2$ and has polynomials of all degrees $k\geq 2$. We analyze each
of the above 3 families in more detail, and then give a proof of Theorem
\ref{thm:X2flag}.

\begin{prop}
  \label{prop:E11flags}
  The $\cE^{(11)}$ flags are classified into the following subclasses,
  according to their degree sequence:
\begin{subequations}\label{eq:E11cases}
  \begin{align}
    \label{eq:E1123dgbasis}
    \cE^{(11)}_{23} &= \cE^{(11)}(a_0, a_1;0,1),\quad\text{with} \quad
a_1a_0+a_1-a_0 \neq
    0\\ \nonumber
    &= \lspan\{ (a_0 a_1 + a_1 - a_0) z^2 + (2 -a_1) (a_0z +1),
     (z-1)^2(1+(2+a_0)z)\}\, \cup \,  \\ \nonumber
    &\qquad\qquad\lspan\{ z^2(z-1)^2 z^j\}_{j=0}^\infty\\
    \label{eq:E1113dgbasis}
    \cE^{(11)}_{13} &= \cE^{(11)}\left(a_0,
\frac{a_0}{1+a_0};0,1\right),\text{ with} \quad a_0\neq -1,\text{ and }
(a_0, a_1) \notin \{(0,0), (-2,2)\}\\ \nonumber &= \lspan\{ a_0 z
    + 1, (z-1)^2(1+(2+a_0)z)\} \, \cup \, \lspan \{ z^2(z-1)^2
    z^j\}_{j=0}^\infty\\
    \label{eq:E1103dgbasis}
    \cE^{(11)}_{03} &= \cE^{(11)}(0,0;0,1)= \lspan\{ 1,
    (z-1)^2(1+2z)\} \,\cup \, \lspan\{ z^2(z-1)^2 z^j\}_{j=0}^\infty\\ 
    \label{eq:E1112dgbasis}
    \cE^{(11)}_{12} &= \cE^{(11)}(-2,2;0,1)= \lspan\{ 2z-1,z^2\} \,\cup \,
\lspan\{
    z^2(z-1)^2 z^j\}_{j=0}^\infty
  \end{align}
\end{subequations}
\end{prop}
\begin{proof}
  This follows by direct inspection of \eqref{eq:E11span}.
\end{proof}
\noindent 
Also note that $\cE^{(11)}_{12}$ can be obtained as a limit of
$\cE^{(11)}_{23}$ by setting $a_0=t-2$, $a_1=t+2$ and then sending
$t= 0$. The flags in Proposition \ref{prop:E11flags} are all $\rX_2$-flags
whose operators have two simple poles at $0$ and $1$. In the following
Proposition we provide a basis for the $\cD_2$-spaces of operators that
preserve them.

\begin{prop}
  \label{prop:E11D2}
  The generic flag $\cE^{(11)}_{23}$ has a 2-dimensional $\cD_2$ space.  The
  non-stable flag $\cE^{(11)}_{13}$ has a
  3-dimensional $\cD_2$, while $\cE^{(11)}_{03}$ and $\cE^{(11)}_{12}$ both have
a 4-dimensional $\cD_2$.  The most general second order operator that preserves
each of these flags is shown below (and therefore a basis of their
$\cD_2$-space). The symbols $a_0, a_1$ denote the flag
  moduli while the symbols $c, c_0, c_1, q_0, \lambda$ denote free constants
appearing in the operator
 {\small
\begin{subequations}
    \begin{align}
      \label{eq:T1123def}
      T^{(11)}_{23}[y] &= c\left(-\frac{1}{2} z^2 (a_0-a_1) (a_0-a_1+4)-z (a_0
        a_1-a_0-a_1^2+3 a_1)-\frac{a_1^2}{2}+a_1\right) y''+\\ \nonumber
      &\qquad +c\left(z ((a_0-a_1)(a_0 a_1 -2 a_0 +2)+2 a_0^2)+(a_0-1)
        a_1^2-(a_0-3) a_1+a_0 (a_0+1)\right) y'+ \\ \nonumber &\qquad
      +\frac{ca_0 (a_0+2)}{z-1}(y'-a_1 y)+\frac{c(a_1-2)
        a_1}{z}(y'-a_0 y) + \lambda y\\
      \label{eq:T1113def}
      T^{(11)}_{13}[y] &= \left( -(c_0+c_1)\frac{z^2}{2}+
        c_0\left(z-\frac{1}{2}\right)\right) y''+\left(\left(a_1c_1- a_0
          c_0\right) z+(a_0-1) c_0+c_1\right) y'+\\ \nonumber &\qquad
      \frac{c_0}{z}(y'-a_0 y)+\frac{c_1}{z-1}(y'-a_1 y) + \lambda
      y,\quad
      a_1 = \frac{a_0}{a_0+1} 
    \end{align}
    \begin{align}
      \label{eq:T1103def}
      T^{(11)}_{03}[y] &=\left( -\left(q_0+c_0 +c_1\right)\frac{z^2}{2}+
        \frac{q_0z}{2}+c_0\left(z-\frac{1}{2}\right)\right)
      y''+\left(q_0 \left(z-\frac{1}{2}\right)-c_0+c_1\right) y'+\\
      \nonumber &\qquad +\left(\frac{c_0}{z}+\frac{c_1}{z-1}\right) y' +
      \lambda
      y\\ \nonumber
      \label{eq:T1112def}
      T^{(11)}_{12}[y] &= \left( (c_0+c_1-q_0)\frac{z^2}{2}+
        \left(\frac{q_0}{2}-c_1\right)z-\frac{c_0}{2}\right)
      y''+\left(q_0 \left(z-\frac{1}{2}\right)-2
        c_0+2 c_1\right) y'+\\
      &\qquad+\frac{c_0}{z}(y'+2y)+\frac{c_1}{z-1}(y'-2y) + \lambda y
    \end{align}
\end{subequations}
  } 
\end{prop}
Before we turn to the proof of this last Proposition, observe the duality
between flag moduli and free parameters in the operator. In the general case
$\cE^{(11)}_{23}$ the flag has two moduli $(a_0,a_1)$ and the $\cD_2$-space has
dimension two. In the case
$\cE^{(11)}_{13}$ the flag has one modulus $a_0$ and the operator has three
free parameters, since $\dim \cD_2\left(\cE^{(11)}_{13}\right)=3$. In the last
two cases $\cE^{(11)}_{03}$ and $\cE^{(11)}_{12}$ the flag is completely
specified (no flag moduli) but the operator contains four free parameters.

\begin{proof}
  By Lemma \ref{lem:Tansatz}, we must consider an operator of the form
  \[ T[y] := p_{-2} y'' + (p_{-1} z y'' + q_{-1} y') + (p_0 z^2 y''+
  q_0 y') + \frac{ c_0(y'-a_0 y)}{z} + \frac{ c_1 (y'-a_1 y)}{z-1},\]
  where $ p_{-2},p_{-1},q_{-1}, p_0, q_{0}, c_1, c_0$ are undetermined
  coefficients that need to be constrained so that $T$ preserves the
  flag in question.  Applying the relation
  \[ y'(0) = a_0 y(0), \quad y\in \cU \] to the constraint
  \begin{equation}
    \label{eq:Ty'0}
    T[y]'(0) - a_0 T[y](0)=0,\quad y\in U
  \end{equation}
  yields:
  \begin{align*}
    &\left(c_0/2+ p_{-2}\right) y'''(0) + \left(3 a_0 p_{-2}-(3 a_0/2)
      c_0-
      c_1+p_{-1}+q_{-1}\right)y''(0)+\\
    &\qquad + \left( a_0^3 c_0+ (a_0^2 - a_0 + a_1) c_1- a_0^2 q_{-1}+
      a_0 q_0\right) y(0)=0
  \end{align*} 
  Since $y'''(0), y''(0), y(0)$ vary freely for $y\in \cU$ the
  coefficients of all 3 terms must vanish in order for \eqref{eq:Ty'0}
  to hold.  An analogous constraint holds for
  \[ T[y]'(1) - a_1 T[y](1)=0.\] Since there are 7 parameters and only
  6 linear, homogeneous constraints, there exists at least one
  non-trivial operator that preserves $\cE^{(1)}$.  The desired
  solution vector \[ [ p_{-2} , p_{-1} , q_{-1} ,p_0 ,q_0, c_1, c_0]^t\]
  belongs to the null-space of the following matrix
  \begin{equation}
    \label{eq:E1ccmatrix}
    \begin{pmatrix}
      1 & 0 & 0 & 0 & 0 & 0 & 1/2 \\
      -a_0 & 1 & 1 & 0 & 0 & -1 & -3 a_0/2 \\
      0 & 0 & -a_0^2 & 0 & a_0 & a_0^2-a_0+a_1 & a_0^3 \\
      1 & 1 & 0 & 1 & 0 & 1/2 & 0 \\
      -a_1 & 1-a_1 & 1 & 2-a_1 & 1 & -3 a_1/2 & 1 \\
      0 & 0 & -a_1^2 & 0 & -(a_1-1) a_1 & a_1^3 & a_0-a_1 (a_1+1) \\
    \end{pmatrix}
  \end{equation}
  A direct calculation shows that all 6 minors of the above $6\times
  7$ have $a_0 a_1 + a_1 - a_0$ as a factor, and that it is not
  possible for all the minors to vanish if $a_0 a_1 + a_1 - a_0 \neq
  0$.  Hence, generically the above constraint matrix has rank 6, and
   there exists a unique, up to a scalar factor, solution, which
  after some calculation provides the operator $T^{(11)}_{23}$.

  Setting $a_1 = a_0/(a_0+1)$ in the above matrix drops the rank of
  the matrix to $5$, provided, $a_0\notin \{ 0,-2\}$.  Now the
  nullspace is 2-dimensional; this gives the form of $T^{(11)}_{13}$.
  Setting $a_0 = a_1 = 0$ in the constraint matrix gives a matrix of
  rank $4$.  The nullspace corresponds to the operator $T^{(11)}_{03}$.
  Similarly, $a_0=-2, a_1=2$ also gives a rank 4 matrix, whose
  nullspace corresponds to the operator $T^{(11)}_{12}$.
\end{proof}

\begin{prop}
  \label{prop:E11X2}
  The flag $\cE^{(11)}_{23}$ is an $\rX_2$ flag, provided $a_0\notin \{
  0,-2\}$ and $a_1\notin \{ 0,2\}$.  The non-stable flags
  $\cE^{(11)}_{13},\cE^{(11)}_{03},\cE^{(11)}_{12}$ are all $\rX_2$ flags.
\end{prop}
\begin{proof}
  It is clear that all the operators preserve codimension two flags and since
they have poles they do not preserve the standard flag. It only remains to
prove the maximality assumption, i.e. that these operators do not preserve a
flag of codimension one.
  By Lemma \ref{lem:1stordcond}, an operator with two poles
  cannot preserve a codimension 1 flag.  By inspection, if
  $a_0, a_1$ satisfy the conditions given above, the operator $T^{(11)}_{23}$
  cannot preserve a codimension 1 flag.  On the contrary, if $a_0 = 0$,
  a direct calculation shows that $T^{(11)}_{23}[1] = 0$ and hence
  $\cE^{(1)}(0,a_1)$ is not the maximal flag preserved by $T^{(11)}_{23}$.
  Similarly, if $a_0 = -2$ then $2z-1$ is again an eigenpolynomial.
  Similar remarks hold for the cases $a_1=0$ and $a_1=2$.

  For the degenerate, non-stable flags, by taking $c_0, c_1\neq 0$ we
  obtain operators that preserve these flags, but have 2 distinct
  poles.  Therefore, by the same argument, these operators cannot
  preserve a flag of smaller codimension.
\end{proof}

We now turn to an analysis of the one-pole $\rX_2$ flag
$\cE^{(2)}$.  In the language of Lemma \ref{lem:ordergaps}, this
flag is the most general codimension $2$ flag with the order sequence
$I_0 = \{ 0,2,4,5,6,\ldots\}$.  The Lemma below derives the constraint
\eqref{eq:E2constraint} as the necessary and sufficient condition for
such a flag to have a non-trivial $\cD_2$.
\begin{lem}
  \label{lem:1poleflag}
  Every $\rX_2$ flag that is preserved by an operator with a unique pole
  is translation-equivalent, to $\cE^{(2)}(a_{01}, a_{03}, a_{23};0)$
  where the parameters satisfy \eqref{eq:E2constraint}. Up to a multiplicative
constant, a second order operator that preserves such a flag has the form
  \begin{equation}
    \label{eq:T3def}
    T^{(2)}[y]= y'' + (p_{-1} z y'' + q_{-1} y') + (p_0 z^2 y''+ q_0
    y') -4 \frac{(y'- a_{01} y)}{z}+\lambda y
  \end{equation}
  where
  \begin{align}
    \label{eq:p-1def}
    p_{-1} & = 2 a_{01} - 2a_{23} \\
    \label{eq:q-1def}
    q_{-1} &= -7 a_{01} + 5 a_{23} 
  \end{align}
  and where $p_0, q_0$ satisfy:
  \begin{equation}
    \label{eq:E2ccmatrix}
    \begin{pmatrix}
      0 & a_{01} & 3 a_{01}^3 - 6 a_{03} - 5 a_{01} a_{23} \\
      2 a_{03} &  a_{03} &  a_{03} (a_{01} -a_{23})(a_{01} + a_{23})\\
      4a_{23} & a_{23} & 6 a_{03} + 5 a_{01} a_{23}^2 - 3 a_{23}^3
    \end{pmatrix}
    \begin{pmatrix}
      p_0 \\ q_0 \\ 1
    \end{pmatrix}
    =
    \begin{pmatrix}
      0\\0\\0
    \end{pmatrix}
  \end{equation}
\end{lem}
\begin{proof}
  By Lemma \ref{lem:Tansatz} an operator with a unique pole at $z=0$
  that preserves a polynomial flag has the form
  \[ T[y]= p_{-2}y'' + (p_{-1} z y'' + q_{-1} y') + (p_0 z^2 y''+ q_0
  y') +c \frac{(y'- a_{01} y)}{z}+\lambda y\] where $p_{-2}, p_{-1},
  q_{-1}, p_0, q_0, \lambda, c$ are undetermined coefficients.  If we
  demand that the flag has codimension $2$, then the flag must be
  $\cE^{(2)}$. By Lemma
  \ref{lem:ordergaps}, it follows that we must also require $T_{-2}[z^5]=0$,
where
  \[ T_{-2}[y] = p_{-2} y'' + \frac{cy'}{z},\]
  This imposes the condition $c = -4 p_{-2}$ and since we require a non-trivial
$T_{-2}$, we   must have $p_{-2}\neq 0$.  Hence, without loss of generality, we
  impose 
  \[ c=-4, \quad p_{-2} = 1 \]
  from here on. The flag $\cE^{2}$ in \eqref{eq:E2def} is defined by the first
and third order conditions 
  \begin{align}
    \label{eq:E3-1stordcond}
    y'(0) &= a_{01} y(0) \\
    \label{eq:E3-3rdordcond}
    y'''(0) &= 6a_{03} y(0) +3 a_{23} y''(0)
  \end{align}
  Imposing these conditions on $T[y]$ yields
  \begin{align*}
    & \left(5 a_{01}-3 a_{23}+p_{-1}+q_{-1}\right)y''(0)+\left(-4
      a_{01}^3-6  a_{03}-a_{01}^2 q_{-1}+a_{01} q_0\right) y(0) =0\\
    &\left(a_{01}+a_{23}+3 p_{-1}+q_{-1}\right) y^{(4)}(0)+ \\
    &\quad +3 \left(-4 a_{01} a_{23}^2+6 a_{03}-6 a_{23}^2 p_{-1}-3
      a_{23}^2 q_{-1}+4 a_{23} p_0+a_{23} q_0\right)y''(0)+\\
    &\quad -6 a_{03} \left(4 a_{01}^2+4 a_{01} a_{23}+6 a_{23}
      p_{-1}+(3 a_{23}+a_{01}) q_{-1}-6 p_0-3 q_0\right)y(0)=0
  \end{align*}
  The values of $y^{(4)}(0), y^{(2)}(0), y(0)$ vary freely for $y\in
  \cE^{(2)}$, and hence, invariance holds if and only if the
  coefficient of each of these expressions vanish.  The conditions
  \eqref{eq:p-1def} \eqref{eq:q-1def} follow from the vanishing of the
  leading order coefficients.  Once these values of $p_{-1}, q_{-1}$
  are imposed, the overdetermined constraint \eqref{eq:E2ccmatrix}
  expresses the vanishing of all the remaining coefficients.  The
  vanishing of the determinant of the matrix in \eqref{eq:E2ccmatrix}
  is the compatibility condition for these constraints, and this is precisely
condition 
  \eqref{eq:E2constraint}.
\end{proof}
\noindent
As we did before for the two-poles $\rX_2$ flags, the one-pole $\rX_2$
flags can be classified according to their degree sequence.
\begin{prop}
  \label{prop:E3flags}
  Every one-pole $\rX_2$ flag is affine-equivalent to one of the
  following:
\begin{subequations}
  \begin{align}
    \label{eq:E2a13span}
    \cE^{(2a)}_{13}(a) &:= \cE^{(2)}(1,0,a;0) = \lspan\{ 1+z
    ,z^2+a z^3,z^4,z^5,\ldots\},\quad a \neq 0,\\
    \cE^{(2a)}_{03} &:= \cE^{(2)}(0,0,1;0) = \lspan\{ 1
    ,z^2+ z^3,z^4, z^5,\ldots\},\\
    \cE^{(2a)}_{12} &:= \cE^{(2)}(1,0,0;0) = \lspan\{ 1+ z
    ,z^2,z^4,z^5,z^6,\ldots\},\\
    \cE^{(2a)}_{02} &:= \cE^{(2)}(0,0,0;0) = \lspan\{
    1,z^2,z^4,z^5,z^6,\ldots\},\\
    \label{eq:E2bdef}
    \cE^{(2b)}_{23}(a) &:= \cE^{(2)}(a,a,a;0) =  \lspan\{
    1+az- z^2,z^2(1+az), z^4, z^5,\ldots\},\quad a\neq 0,\\  
    \cE^{(2c)}_{23}(a) &:=  \cE^{(2)}(a,-a(a+1)/6,1),\quad a\neq 0\\ \nonumber
    &= \lspan \{ 1+az + a(a+1)z^2/6, z^2+z^3, z^4, z^5,\ldots \}
  \end{align}
\end{subequations}
\end{prop}
\begin{proof}
  The three types of flags labelled $(2a)$, $(2b)$ and $(2c)$ correspond to the
three different ways of satisfying the defining constraint
\eqref{eq:E2constraint} on the three flag moduli.
  The type $(2a)$ flags are obtained by applying the constraint
  $a_{03} = 0$ to a general type $\cE^{(2)}$ flag.  By
  \eqref{eq:E2span}, the resulting degree regular basis is
  \[ 1+a_{01} z, z^2 + a_{23} z^3, z^4, z^5, \ldots \] If $a_{01},
  a_{23} \neq 0$, then a scaling transformation can be used to send
  one (but not both) of the above parameters to $1$.  The various
  subclasses listed above arise if one or both of  $a_{01}, a_{23}=
  0$.

  The type $(2b)$ flag is obtained by applying the constraint $a_{23}
  = a_{01}$.  An examination of \eqref{eq:E2ccmatrix} shows that it is
  not possible for $a_{23} = a_{01} = 0, a_{03} \neq 0$.  Therefore,
  for the type $(2b)$ subcase, we must have $a_{01} \neq 0$.  Thus, in
  this case, transforming \eqref{eq:E2span} to a degree regular basis
  gives
  \[ 1+a_{01} z - \frac{a_{03}}{a_{01}} z^2, z^2 + a_{01} z^3, z^4,
  z^5,\ldots \] Now a scaling transformation can be used to set
  $a_{03}/a_{01}= 1$.

  The type $(2c)$ flags are obtained by imposing
  \[a_{03} = -a_{01} a_{23}(a_{01} + a_{23})/6.\] In this case, the
  degree regular basis is
  \[ 1+a_{01} z + a_{01}(a_{01} + a_{23}) z^2/6, z^2(1 + a_{23} z),
  z^4, z^5,\ldots \] No generality is lost if we assume that $a_{01},
  a_{23} \neq 0$, because otherwise we will obtain a flag of type
  $(2a)$.  Finally, a scaling transformation is used to set $a_{23}
  = 1$.
\end{proof}
\noindent Note: as above the flag subscript indicates the degree
sequence of the flag.
\begin{prop}
  \label{prop:E2D2}
  The flags $\cE^{(2a)}_{13}, \cE^{(2b)}_{23}, \cE^{(2c)}_{23}$ have a
  2-dimensional $\cD_2$.  The degenerate flags
  $\cE^{(2a)}_{03},\cE^{(2a)}_{12}$ have a 3-dimensional $\cD_2$,
  while $\cE^{(2a)}_{02}$ has a 4-dimensional $\cD_2$.  The most
  general second order operator that preserves each of these flags is
  shown below.  The symbol $a$ represents the flag modulus while the
  symbols $c,p_0, q_0, \lambda$ are free constants that appear in the
  operator. 

  {\small
  \begin{subequations}  
    \begin{align}
      \label{eq:T2a13def}
      T^{(2a)}_{13}[y] &= c\left( \left(1-3 a\right)
        \left(3-a\right)\frac{z^2}{4} +2 \left(1-a\right)z+1\right)
      y''+\\ \nonumber &\qquad + c\left( \left(5 a-3 \right) az+5
        a-7\right) y' +\frac{4 c(y-y')}{z}+\lambda y\\
      \label{eq:T2a03def}
      T^{(2a)}_{03}[y] &= \left(
        \left(3c-q_0\right)\frac{z^2}{4}+c(1-2 z)\right)y''+ 
      \left(5c+q_0 z\right)y'-\frac{4c y'}{z}+\lambda y\\
      \label{eq:T2a12def}
      T^{(2a)}_{12}[y] &=\left(p_0 z^2+c(2 z +1)\right)y''-c\left(3 z
        +7\right) y'+\frac{4 c( y-
        y')}{z}+\lambda y\\
      \label{eq:T2a02def}
      T^{(2a)}_{02}[y] &= \left(p_0 z^2+c\right) y''(z)+q_0 z
      y'(z)-\frac{4 cy'(z)}{z}+\lambda y \\
      \label{eq:T2b23def}
      T^{(2b)}_{23}[y] &= c\left(1-z^2
        \left(a^2+ 3\right)\right) y''+ c\left(2 z 
        \left(a^2+ 3\right)-2 a\right)
      y'(z) +\frac{4c (a y- y')}{z}+\lambda y\\
      \label{eq:T2c23def}
      T^{(3c)}_{23}[y] &= c(1+(a-1)z)^2
      y'' +c((a-1)(1-3a) z+5-7a) y'+\frac{4c (a y- y')}{z}+\lambda y
    \end{align}
  \end{subequations} }
\end{prop}


\begin{proof}
  Each of the flags in question is a specialization of the $\cE^{(2)}$
  flag discsussed in Lemma \ref{lem:1poleflag}, imposed in such a way
  so that \eqref{eq:E2constraint} holds.  The 3 factors in
  \eqref{eq:E2constraint} give us the 3 possible cases: $\cE^{(2a)},
    \cE^{(2b)}, \cE^{(2c)}$.  Imposing the respective constraints
  \[ a_{03} = 0,\quad a_{23} = a_{01},\quad a_{03} = -a_{01} a_{23}
  (a_{01}+a_{23})/6 \] transforms \eqref{eq:E2ccmatrix} into a
  consistent, rank 2 system.  We can further eliminate one more
  parameter by means of an appropriate scaling transformation. The
  form of the operators shown above follows from \eqref{eq:p-1def}
  \eqref{eq:q-1def} and the solution of the corresponding
  \eqref{eq:E2ccmatrix}.
\end{proof}

\begin{prop}
  \label{prop:E2X2}
  The flags
  $\cE^{(2a)}_{13},\cE^{(2a)}_{03},\cE^{(2a)}_{12},\cE^{(2a)}_{02},\cE^{(2b)}_{23},
  \cE^{(2c)}_{23}$ are 
  all $\rX_2$ flags.   
\end{prop}
\begin{proof} 
  For each of the above flags, we have exhibited a singular
  operator that preserves it. It remains to show that these operators
  cannot preserve a flag of smaller codimension.  By Lemma
  \ref{lem:ordergaps}  an $\rX_1$ flag preserved by an operator with
  a pole at $z=0$, must have elements of order $0,2,3,4,\ldots$.
  Therefore, it suffices to check that $T_{-2}$ (see the Lemma for the
  explanation of the notation) does not annihilate $z^3$.  For each of
  the  operators shown in the preceding Proposition,
  \[ T_{-2}[y] = y''-\frac{4y'}{z}.\]
  Hence,
  \[ T_{-2}[z^3] = -6 z .\]
  Therefore,  none of these operators can preserve an $\rX_1$ flag.
\end{proof}

\begin{proof}[Proof of Theorem \ref{thm:X2flag}]
  By the above Lemmas, an $\rX_2$ operator  has either one or two
  poles.  In the last case, the corresponding $\rX_2$ flag
  satisfies two distinct first order condtions
  \[ y'(b_i) = a_i y(b_i), \quad i=1,2 \] Applying an affine
  transformation, no generality is lost if we assume that the poles
  are at $z=0$ and $z=1$. This gives us flags of type $\cE^{(11)}$.
  The corresponding $\rX_2$ operators are given in Proposition
  \ref{prop:E11D2}.  The $\rX_2$ assertion is verified in Proposition
  \ref{prop:E11X2}.

  In the case of one pole, without loss of generality the pole is at
  $z=0$.  In this case, the flag satisfies a first and a third order
  condition, which gives us a flag of type $\cE^{(2)}$.  As it was shown
  in Lemma \ref{lem:1poleflag}, the moduli of the general flag must
  satisfy the constraint \eqref{eq:E2constraint}.  This gives us the
  three cases: $\cE^{(2a)}, \cE^{(2b)}, \cE^{(2c)}$.  The
  corresponding operators for these flags are given in Proposition
  \ref{prop:E2D2} and the $\rX_2$ condition is verified in Proposition
  \ref{prop:E2X2}
\end{proof}

\section{Factorization of exceptional operators}
\label{sect:fact}
The results in this section are concerned with factorizations of the
differential operators that preserve $\rX_2$ flags and
their connection to the Darboux transformation. The usual Darboux
transformation involves Schr\"odinger operators and square-integrable
eigenfunctions but for our purposes it will be convenient to generalize it to
second order operators with rational coefficients.

\begin{definition}
  \label{def:ratfact}
 Let $T$ be a second order differential operator that preserves a
polynomial flag $\cU$. Let
  \begin{equation}
    \label{eq:TBA}
    T = B A+ \lambda_0
  \end{equation}
 be a factorization of $T$ where $A,B$ are first
  order operators with rational coefficients and $\lambda_0$ is a constant. If
the partner operator defined by
  \begin{equation}
    \label{eq:hTAB}
    \hT= A B  +\lambda_0.
  \end{equation}
 also preserves a polynomial flag $\hat \cU$ we will say that $T$ and $\hat T$
are related by an \textit{algebraic Darboux transformation}.
\end{definition}
\begin{definition}
  More generaly, we will say that two operators $T$ and $\hT$ are
\textit{Darboux-connected} if there exists a sequence of algebraic Darboux
transformations that connect them.
\end{definition}
%
%

The same notion can be defined for polynomial flags in the following manner:
\begin{definition} \label{def:darboux-connected}
 Two polynomial flags $\cU:U_1\subset
U_2\subset\dots$ and $\hat{\cU}:\hat U_1\subset \hat U_2\subset\dots$ are
  \emph{Darboux-connected}  if there exists two first order
  rational operators $A$ and $B$ such that one of the following three
  possibilities occur:
  \begin{align}
    \label{eq:flagiso}
    A[U_i]&\subset\hU_i,  &B[\hU_i]&\subset U_i,\quad i\geq 1;\\
    \label{eq:flagstatedel}
    A[U_{i+1}]&\subset\hU_i,  &B[\hU_i]&\subset U_{i+1},\quad i\geq 1,&
    A[U_1] = 0;\\
    \label{eq:flagstateadd}
    B[U_{i+1}]&\subset\hU_i,  &A[\hU_i]&\subset U_{i+1},\quad i\geq 1,& B[U_1]
    = 0.
  \end{align}
  In accordance with \cite{GKM10} we will refer to the above cases as
   formally isospectral, formally state-deleting and
  formally state-adding.
\end{definition}

Note that this implies that the second order operators $T=BA$ and $\hat T=AB$
preserve the flags $\cU$ and $\hat\cU$ respectively, so Darboux-connected polynomial flags are always invariant.
 It is common to refer to
the operators $A,B$ as \emph{intertwining operators}, or simply as
\emph{intertwiners}.
\begin{definition}
  We say that a polynomial flag $\cU$ is an $m$-step flag if there exists a
sequence of $m$ Darboux transformations that connect $\cU$ to the standard flag.
\end{definition}


Our main results in this section are summarized in the following two theorems:
\begin{thm}
  \label{thm:X12flagstep}
  Every $\rX_1$ flag is a 1-step flag.
  Every $\rX_2$ flag is either a 1-step or a 2-step flag.
\end{thm}

\begin{thm}
  \label{thm:ExDarb}
  Every $\rX_1$ and $\rX_2$ operator is Darboux-connected to a
  classical operator.  Furthermore, the intertwining operators that
  connect the classical operator to the $\rX$-operator  also
  connect the standard flag to the exceptional flag.
\end{thm}
\noindent 
As we show in the next section, one consequence of Theorem
\ref{thm:ExDarb} is that all $\rX_2$ and $\rX_1$ orthogonal polynomials
can be expressed as certain Wronskians involving classical OPs.

Using the classification of $\rX_1$ and $\rX_2$ flags from the
preceding section, 
the proof of Theorem \ref{thm:X12flagstep} is broken up into 
a series of  Lemmas.
It turns out that Theorem \ref{thm:ExDarb} is a consequence of Theorem
\ref{thm:X12flagstep}.  Our proof strategy is to show that if two
polynomial flags are Darboux-connected, then so are the operators that
preserve them.  This fact is established by Lemmas
\ref{lem:X1Darboux}, \ref{lem:statedel}  and \ref{lem:dim2Darboux}.  We
complete the proof of Theorem \ref{thm:ExDarb} at the end of the
present section.

\begin{lem}
  \label{lem:1stepX1}
  Every $\rX_1$ polynomial flag is a 1-step flag.
\end{lem}
\begin{proof}
  Let $\cU=\cE(a;b)$ be an $\rX_1$ flag as per Theorem
  \ref{thm:X1flag}. Without loss of generality, $b=0$.  Define the 1st
  order operators
  \begin{equation}
    \label{eq:X1ABdef}
    A[y]:= \frac{y'-ay}{z},\quad
    B[y] := z y'-(az+1) y
  \end{equation}
  By inspection, 
  \[A[U_i] \subset \cP_{i-1},\quad i=1,2,\ldots\]
  Also, 
  \[ B[y]'(0)  - a B[y](0) = y'(0)- ay(0)-y'(0)+ a y(0)= 0 \]
  Hence,
  \[ B[\cP_{i-1}] \subset U_{i},\quad i=1,2,\ldots. \] Therefore, $AB$
  preserves the standard flag, while $BA$ leaves invariant $U_i$ for
  every $i=1,2,\ldots$ .
\end{proof}
\begin{lem}
  \label{lem:X1Darboux}
  Every $\rX_1$ operator is Darboux-connected to a classical operator.
\end{lem}
\begin{proof}
  Let $\cU=\cE(a;b)$ be an $\rX_1$ flag as per Theorem
  \ref{thm:X1flag}. Without loss of generality, $b=0$.  Let
  \[ A_{\al_1}[y] = A[y] + \al_1 y',\quad B_{\al_2}[y] = B[y] + \al_2 z^2 y',\]
  where $A,B$ are
  the operators defined in \eqref{eq:X1ABdef}.  Observe that
  \begin{equation}
    \label{eq:X1Darboux2}
    A_{\al_1}[U_i] \subset \cP_{i-1},\qquad 
    B_{\al_2}[\cP_{i-1}] \subset U_{i},\quad i=1,2,\ldots.
  \end{equation}
  and that
  \begin{equation}
    \label{eq:X1Darboux1}
    \dim \{ c B_{\al_2} A_{\al_1} + \lambda : \al_1, \al_2, c,\lambda \in \Rset
    \} = 4.
  \end{equation}
  It follows that every operator in the vector space in
  \eqref{eq:X1Darboux1} preserves the $\rX_1$ flag.  In
  \cite[Proposition 4.10]{GKM6} it was shown that $\dim \cD_2(\cU) = 4$.
  Therefore, by dimensional exhaustion, every operator $T\in\cD_2(\cU)$
  admits a rational factorization of the form $T = c B_{\al_2} A_{\al_1} +
  \lambda$.  To conclude, we observe that, by \eqref{eq:X1Darboux2}, the
  partner operator $\hT=cA_{\al_1}B_{\al_2}+\lambda$ preserves the standard
  polynomial flag.
\end{proof}

\begin{lem}
  \label{lem:dim2D2}
  Let $\cU$ be a polynomial flag. If $\dim\cD_2(\cU)\geq 2$ then there
  exists a second order operator $T\in \cD_2(\cU)$.  If $\dim
  \cD_2(\cU) = 2$, exactly, then $\{1,T\}$ is a basis of $\cD_2(\cU)$.
\end{lem}
\begin{proof} 
  It is clear that $1\in\cD_2(\cU)$. If there exists a first order
  operator $S\in\cD_2(\cU)$, then $S^2\in\cD_2(\cU)$ is a second order
  operator, as was to be shown. It also follows that, if $\cD_2(\cU)$
  contains an operator of the 1st order, then $\dim\cD_2(\cU)\geq 3$.
  Hence, if $\dim \cD_2(\cU) = 2$, exactly, then every
  $T\in\cD_2(\cU)$ is either a constant multiplication operator or an
  operator of the second order.
\end{proof}

\begin{lem}
  \label{lem:statedel}
  Let $\cU\subset \cP$ be a polynomial flag and $A[y]$ a 1st order
  operator such that $\hcU:= A[\cU] \subset \cP$ is also a polynomial
  flag.  Furthermore, suppose that $A[U_1] = \{ 0\}$ and that $\dim
  \cD_2(\cU) \geq 2$.  Then, $\cU, \hcU$ are Darboux
  connected. Furthermore, every operator in $\cD_2(\cU)$ is
  Darboux-connected to an operator in $\cD_2(\hcU)$.
\end{lem}
\begin{proof}
  Choose a non-zero $\phi\in U_1$.  Let $T\in \cD_2(\cU)$ be given.
  Since $\phi$ spans $U_1$ and since $T[U_1]\subset U_1$ we must have
  \[ (T-\lambda)[\phi] = 0 \] for some $\lambda\in \Rset$.  Write
  \begin{align*}
    T[y] &=py''+q y'+ry\\
    A[y] &= b(y'-wy)
  \end{align*}
  where $p(z),q(z),r(z),b(z)$ are rational functions and where
  $w(z)=\phi'(z)/\phi(z)$, because $A[\phi]=0$, as per the above
  assumption.  Next, set
  \[    B[y] = \hb(y'-\hw y) \]
  where
  \[    \hw = -w-q/p+b'/b,\quad    \hb= p/b \]
  A direct calculation then shows that
  \[ T = BA + \lambda.\] 
  Since the kernel of $A|U_{i+1}$ is 1-dimensional we actually
  have
  \[ \hU_i = A[U_{i+1}],\quad i=1,2,\ldots \]
  Since
  \[ T[U_i] \subset U_i,\quad i=1,2,\ldots\] it follows that
  \[ B[\hU_i] \subset U_{i+1},\quad i=1,2,\ldots\] Therefore $AB\in
  \cD_2(\cU)$ and $BA\in \cD_2(\hcU)$.  By Lemma \ref{lem:dim2D2},
  there exists a $T\in \cD_2(\cU)$ such that $p(z)\neq 0$.  This
  proves that $\cU$ and $\hcU$ are Darboux connected.
\end{proof}

\begin{lem}
  \label{lem:dim2Darboux}
  Let $\cU,\hcU$ be Darboux-connected polynomial flags.  If
  $\dim\cD_2(\cU)=2$, then every operator in $\cD_2(\cU)$ is
  Darboux-connected to an operator in $\cD_2(\hcU)$.
\end{lem}
\begin{proof}
  Let $A[y]$ and $B[y]$ be 1st order operators that connect the two
  flags. It is clear that $T=c BA + \lambda$ preserves $\cU$ for all
  $c,\lambda\in \Rset$. By exhaustion every operator in $\cD_2(\cU)$
  has this form. By assumption, the partner operator $\hT=c
  AB+\lambda$ preserves the partner flag $\hcU$.
\end{proof}

\begin{lem}
  \label{lem:1stepE1123}
  The flag $\cE^{(11)}_{23}$ is a 1-step flag.
\end{lem}
\begin{proof}
  Recall that $\cE^{(11)}_{23} = \cE^{(1)}(a_0,a_1;0,1)$ where $a_0
  a_1+a_1-a_0\neq 0$.
  Consider the 1st order operators
  \begin{align*}
    A[y] &= a_1 \frac{y' - a_0 y}{z} -  a_0 \frac{y' -  a_1 y}{z-1}\\
    B[y] &= z(z-1)(2-a_1+(a_1-a_0-4)z)\,y' + \\
    &\qquad + ((a_0 a_1+a_1-a_0)
    z^2+(2-a_1) a_0 z + 2-a_1)\,y 
  \end{align*}
  Let $U_1\subset U_2 \subset \ldots $ be the flag corresponding to
  the total space $\cE^{(11)}_{23}$; see \eqref{eq:E1123dgbasis} for
  a degree regular basis.  A direct calculation shows that
  \[ B[y]'(0) = a_0 B[y](0),\quad B[y]'(1) = a_1 B[y][1] \]
  Since $B$ raises degree by $2$, it follows that 
  \[ B[\cP_{j-1}] \subset U_j,\quad j=1,2,\ldots \] From the
  definition \eqref{eq:E1def}, we see that $A[\cE^{(11)}_{23}] \subset
  \cP$.  Furthermore,
  \begin{align*}
    A[z^j] &=
    \frac{((a_1-a_0) j+a_0 a_1) z^{j-1} - z^{j-2} j a_1}{z-1} \\
    &= \frac{(a_1-j) a_0}{z-1}+ (a_0 a_1 + j(a_1-a_0)) z^{j-2} +
    \text{lower deg. terms}
  \end{align*}
  Since $\deg U_{j} = j+1$, it follows that
  \[ A[U_j] \subset \cP_{j-1},\quad j=1,2,\ldots \]
  as was to be shown.
\end{proof}
\begin{lem}
  \label{lem:1stepE113}
  The flag $\cE^{(11)}_{13}$ is a 2-step flag.
\end{lem}
\begin{proof}
  The degree regular basis is shown in \eqref{eq:E1113dgbasis}.  In
  particular, 
  \[ U_1 = \lspan \{ 1+a_0 z\} .\] Define
  \[ A[y]:=  \frac{a_1 \cW[y,1+a_0 z]}{z(1-z)} =
  a_1 \frac{y' - a_0 y}{z} -  a_0 \frac{y' -  a_1 y}{z-1},\quad a_1
  = \frac{a_0}{1+a_0}  \]
  A direct calculation shows that
  \[ A[y]'	\left(-1/a_0\right) - a_1 (2+a_0) A[y]\left(-1/a_0\right) = 0,\quad a_1 =
  \frac{a_0}{1+a_0} \]
  Hence 
  \[ A[\cE^{(2)}_{13}] = \cE^{(1)}\left(-a_1(2+a_0); -1/a_0\right) \] The latter
  is an $\rX_1$ flag, and $\rX_1$ flags are 1-step.  Therefore, the
  desired conclusion follows by Lemma \ref{lem:statedel}.
\end{proof}

\begin{lem}
  \label{lem:1stepE103}
  The flag $\cE^{(11)}_{03}$ is a 1-step flag.
\end{lem}
\begin{proof}
  Define 
  \[ A[y] :=  \frac{y'}{z(z-1)} \]
  Using \eqref{eq:E1103dgbasis}, a direct calculation shows that 
  \[ A[\cE^{(11)}_{03}] = \cP \]
where the last equality should be understood as an equality between polynomial
flags.  The desired conclusion follows by Lemma \ref{lem:statedel}
\end{proof}

\begin{lem}
  \label{lem:1stepE112}
  The flag $\cE^{(11)}_{12}$ is a 2-step flag.
\end{lem}
\begin{proof}
  The degree regular basis is shown in \eqref{eq:E1112dgbasis}.  In
  particular, note that \[U_1 = \lspan \{ 2 z-1\}.\] Define
  \[ A[y]:=  \frac{a_1 \cW[y,2 z-1]}{z(1-z)} \]
  A direct calculation shows that
  \[ A[y]'(1/2) = 0.\]
  Hence 
  \[ A[\cE^{(2)}_{12}] = \cE^{(0)}(0, 1/2) \]
  The latter is an $\rX_1$ flag, and $\rX_1$ flags are
  1-step.  Therefore, the desired conclusion follows by \ref{lem:statedel}.
\end{proof}

\begin{lem}
  \label{lem:2stepE2b}
  The flag $\cE^{(2b)}_{23}(a)$ is a 2-step flag.
\end{lem}
\begin{proof}
  Define the operator
  \[ A[y] := (y'-a y)/z + K y',\quad K = \sqrt{a^2\pm 3}\] Applying
  $A$ to the degree regular basis shown in \eqref{eq:E2bdef} gives a
  flag with a stable degree sequence of $1,2,\ldots$.  Imposing
  \[ y'(0) = a y(0),\quad y'''(0) = 3 a y''(0) \pm 6 a
  y(0), \]
  a direct calculation shows that
  \[ A[y]'(0) = (a+K) A[y][0].\] Since the former conditions defines
  $\cE^{(2b)}$ and since the latter conditions defines
  $\cE^{(1)}(a+K;0)$ (see \eqref{eq:EX1def} for the definition), it
  follows that
  \[ A[\cE^{(2b)}] \subset \cE^{(1)}(a+K;0) \]

  Next, define
  \[ B[y]:= z(1-Kz) y' - (3+(a-2K)z) y \]
  If we suppose that
  \[ y'(0) = (a +K) y(0) \]
  then by direct calculation, 
  \[ B[y]'(0) = a B[y](0),\quad  B[y]'''(0) = 3 a y''(0)+ 6 a y(0) \]
  Therefore,
  \[ B[\cE^{(1)}(a+K;0)] \subset \cE^{(2b)}.\]
  
  Next, observe that $A$ lowers degree by $1$ and that $B$ raises
  degree by $1$.    Hence $BA$ and $AB$ do not raise degree and they
  preserve their respective flags.
  Since $\cE(a+K;0)$
  is a 1-step flag (Theorem \ref{thm:X1flag} and Lemma
  \ref{lem:1stepX1}) it follows that $\cE^{(3b)}_{23}$ is a 2-step flag.
\end{proof}

\begin{lem}
  \label{lem:2stepE2c}
  The flag $\cE^{(2c)}_{23}(a)$ is a 2-step flag.
\end{lem}
\begin{proof}
  The argument is the same as for the proof of Lemma
  \ref{lem:2stepE2b}, but with the following operators:
  \begin{align*}
    A[y] &:= \frac{y'-a y}{z} + \frac{a-1}{2} y'\\
    B[y] &:= z(1+(a - 1) z) y' - (3+(2a -1)z) y.
  \end{align*}
  We then have
  \[ A[\cE^{(2c)}]\subset \cE^{(0)}(1;0),\quad B[\cE^{(0)}(1;0)]
  \subset \cE^{(2c)}.\]
\end{proof}

\begin{lem}
  \label{lem:2stepE2a}
  The flags $\cE^{(2a)}_{13}, \cE^{(2a)}_{03}, \cE^{(2a)}_{12} ,
  \cE^{(2a)}_{02}$ are all 2-step flags.
\end{lem}
\begin{proof}
  By Proposition \ref{prop:E3flags}, all of the above flags are
  various specializations of 
  \[ \cE^{(2)}(a_{01},0, a_{23};0) = \lspan\{ 1+ a_{01} z, z^2 + a_{23}
  z^3, z^4, z^5,\ldots \} .\] Hence, it suffices to prove the
  assertion for this general case.  Equivalently, the above flag
  consists of polynomials satisfying
  \begin{equation}
    \label{eq:E3diffcond}
    y'(0) = a_{01} y(0), \quad  y'''(0) =     3 a_{23} y''(0)
  \end{equation}

  Consider the operator
  \[ A[y] := \frac{y'- a_{01} y}{z} + a_{01} y' \] and note that
  \[ A[a_{01} z + 1] = 0. \] 

  Next, observe that
  \begin{align*}
    A[y]'(z) - \frac{1}{2}(a_{01}+3 a_{23}) A[y](z) &= (a_{01} y(0)
    - y'(0)) \left( \frac{1}{z^2} + \frac{(a_{01} + 3
        a_{23})/2}{z}\right) +\\
    &+\qquad    \frac{1}{2}(y'''(0) - 3 a_{23} y''(0))+ O(z) 
  \end{align*}
  Hence, if $y(z)$ satisfies \eqref{eq:E3diffcond}, then $A[y] \in
  \cE^{(1)}((a_{01} + 3 a_{23})/2;0)$.  

  At this point, let us suppose that $a_{01} \neq 0$ and note that
  \begin{align*}
    A[y]'(z) - a_{01} A[y](z) &= (1+a_{01} z) \left(
      \frac{a_{01}}{z^2} y - \frac{1+a_{01} z}{z^2} y'+ \frac{1}{z}y''\right)
  \end{align*}
  Hence $A[y] \in \cE^{(1)}(a_{01};-1/a_{01})$ for all polynomials
  $y(z)$.  Together, the above calculations demonstrate that if
  $a_{01}\neq 0$, then
  \[ A[\cE^{(2a)}] \subset \cE^{(11)}((a_{01} + 3 a_{23})/2,a_{01} ;0,
  -1/a_{01}).\] Hence, by Lemma \ref{lem:statedel}, the flags
  $\cE^{(2a)}_{13}, \cE^{(2a)}_{12}$ are Darboux connected to the flag
  above.  We already showed that $\cE^{(11)}$ is a 1-step flag, so this
  concludes the proof for the case $a_{01} \neq 0$.

  Finally, let us consider the case $a_{01}=0$.  In this case,
  \begin{align*}
    A[y] &= \frac{y'}{z},\\
    \cE^{(2)}(0,0,a_{23};0) &=  \lspan \{  1, z^2+a_{23} z^3,z^4, z^5,\ldots \}
\\
    A[\cE^{(2)}(0,0,a_{23};0)] &= \lspan \{ 2 + 3a_{23} z, z^2, z^3,\ldots \} \\
    &= \cE^{(1)}(3a_{23}/2;0)
  \end{align*}
  By Lemma \ref{lem:1stepX1}, the latter is a 1-step flag. Since
  $A[1]=0$, applying Lemma \ref{lem:statedel} shows that
  $\cE^{(2a)}_{03}, \cE^{(2a)}_{02}$ are both $2$-step flags.
\end{proof}

\begin{proof}[Proof of Theorem \ref{thm:ExDarb}]
  There are two basic mechanisms which we use to give the proof of the
  conjecture for $\rX_2$ and $\rX_1$ operators.  The first mechanism
  is that of dimensional exhaustion, and is utilized in Lemma
  \ref{lem:X1Darboux} and in Lemma \ref{lem:dim2Darboux}.  This
  mechanism is used to prove the conjecture for $\rX_1$ flags (Lemma
  \ref{lem:X1Darboux}) and also used in the proof of Lemmas
  \ref{lem:1stepE1123}, \ref{lem:2stepE2b} and \ref{lem:2stepE2c}.  All
  these cases require that we exhibit both an $A$ operator, which
  relates the given flag $\cU$ to a ``simpler'' flag $\hcU$, and a $B$
  operator that relates $\hcU$ back to $\cU$.

  The other basic argument is conceptually related to state-deleting
  transformations in quantum mechanics.  Here it suffices to show that
  a 1st order operator that annihilates $U_1$ maps the given flag
  $\cU$ to a simpler flag $\hcU$ and to have in hand a second order
  operator that preserves the given $\cU$.  This is the argument of
  Lemma \ref{lem:statedel}.  This argument is utilized in Lemmas
  \ref{lem:1stepE113}, \ref{lem:1stepE103} and \ref{lem:1stepE112}
  \ref{lem:2stepE2a}.  Taken together, these Lemmas cover the cases of
  all possible $\rX_1$ and $\rX_2$ flags and the operators that
  preserve them.
\end{proof}

\section{Polynomial Sturm-Liouville problems and Darboux transformations}

Our main goal is to complete the classification of $\rX_2$ OPS and what
 remains to do is to select from all the $\rX_2$ operators given in
Section \ref{sect:Classif} for each $\rX_2$ flag, those that give rise to a well
defined Sturm Liouville problem. For this reason, in this Section we
need to review some preliminary results from the theory of
Sturm Liouville problems. We will also provide the main definitions and
properties of algebraic Darboux transformations for second order differential
operators. We emphasize that by construction these transformations will map an
SL-OPS into an SL-OPS.

\subsection{Orthogonal polynomials on the real line defined by a Sturm-Liouville
  problem}
Every second-order eigenvalue equation
\[  T[y]:= p(z) y''+ q(z) y' + r(z) y = \lambda y\]
can be put into formal Sturm-Liouville form
\[ - (Py')' + R y = -\lambda W y \]
where
\begin{align}
  \label{eq:Pdef}
  P(z) &= \exp\left(\int^z  \!\!q/p\,dz\right),\\
  \label{eq:Wdef}
  W(z) &= (P/p)(z),\\
  \label{eq:Rdef}
  R(z) &= -(rW)(z),
\end{align}
With the above definitions, the operator $T[y]$ is formally
self-adjoint with respect to the weight $W(z) dz$ in the sense that
Green's formula, below, holds:
\begin{equation}
  \label{eq:Msymmetric}
  \int  T[f]g \, Wdz - \int T[g] f\, W dz=  P(f'g-fg')
\end{equation}
If the operator $T[y]$ has infinitely many polynomial eigenfunctions,
and if an interval of orthogonality can be appropriately chosen so
that $W(z) dz$ has finite moments and  the right-hand side of
\eqref{eq:Msymmetric} vanishes for polynomials $f(z), g(z)$, then the
eigenpolynomials of $T[y]$ constitute an SL-OPS.

By direct inspection, every $\rX_2$ operator listed in Propositions
\ref{prop:E11D2} and \ref{prop:E2D2} has the form
\[ T[y] := p(z) (y''-2(\log \xi)') + q(z) y' + r(z) y \] where $p(z)$
is a quadratic polynomial, $q(z)$ is a linear form, $\xi(z)$ is either
$z(z-1)$ or $z$ and  $r(z)$ is a rational
function with $\xi(z)$ in the denominator.
Applying an affine change of variable,
\[ z = ax + b \] the coefficients $p(z)$ and $q(z)$ can be put into a
normal form.  There are five classes of these normal forms, which we display in Table 1 together with the interval of orthogonality and the
 weight defined by \eqref{eq:Pdef} -\eqref{eq:Wdef}.
\begin{table}[h]\label{tab2}\caption{}
  \centering
  \begin{tabular}{c|c|c|c|c}
    $p(x)$ & $q(x)$ & $W(x)$ & $I$ &  OPS family \\  \hline
    $1$ & $- 2x$ &$\frac{e^{- x^2}}{\xi(x)^2}$ & $(-\infty,\infty)$
&Hermite\\[10pt]
    $x$ & $\alpha+1- x$ & $\frac{e^{-x} x^\alpha}{ \xi(x)^2}$ &
    $(0,\infty)$& Laguerre\\[10pt]
    $1-x^2$ & $\beta-\alpha-(2+\alpha+\beta) x$ & $\frac{(1-x)^\alpha
    (1+x)^\beta}{\xi(x)^2}$ & $(-1,1)$& Jacobi\\[10pt]
    $x^2$ & $2(x\pm 1)$ & $\frac{e^{\mp 2/x}}{\xi(x)^2}$ & n/a & Bessel\\[10pt]
    $1+x^2$ & $\alpha+ 2(\beta+1)x$ & $\frac{(1+x^2)^\beta e^{a \tan^{-1}
      x}}{\xi(x)^2}$ & n/a &twisted Jacobi 
  \end{tabular}
\end{table}

Just as in the analysis of classical orthogonal polynomial systems
\cite{Lesky} the Bessel and twisted Jacobi cases can be excluded
because it is not possible to choose a interval of orthogonality that
satisfies the finite-moment condition.  Therefore the search for
$\rX_2$ orthogonal polynomial systems narrows to the first 3 cases.  In
each case, the requirement is that $\xi(z)$ have no zeros on the
corresponding interval of orthogonality.  For the Laguerre subcase, there
is the additional constraint that $\alpha>-1$.  For the Jacobi
subcase, the constraint is that $\alpha,\beta>-1$.

%

\subsection{Factorization and orthogonal polynomials}
Consider two differential operators:
\begin{align} 
  T[y]&=py''+q y'+ry,\\
  \hT[y] &= py'' + \hq y' + \hr y,
\end{align}
related by a factorization \eqref{eq:TBA} \eqref{eq:hTAB}. Let us
write
\begin{align}
  \label{eq:Adef} A[y] &= b(y'-wy),\\
  \label{eq:Bdef} B[y] &= \hb(y'-\hw y),
\end{align} where $p(z),q(z),r(z),b(z),w(z),\hb(z),\hw(z)$ are all
rational functions.  We will refer to
  \begin{equation}
    \label{eq:phidef}
    \phi(z) = \exp\int^z\!\! w\,dz,\quad w=\phi'/\phi
  \end{equation}
  as a \emph{quasi-rational factorization eigenfunction} and to $b(z)$
  as the \emph{factorization gauge}.
The reason for the above terminology is as follows.  By \eqref{eq:TBA},
\begin{equation}
  \label{eq:phi0rel}
  T[\phi] = \lambda_0 \phi;
\end{equation}
hence the term factorization eigenfunction.  Next, consider two
factorization gauges $b_1(z), b_2(z)$ and let $\hT_1[y], \hT_2[y]$ be
the corresponding partner operators.  Then,
\[ \hT_2 = \mu^{-1} \hT_1\mu,\quad \text{where } \mu(z) = b_1(z)/
b_2(z).\] Therefore, the choice of $b(z)$ determines the gauge of the
partner operator. This is why we refer to $b(z)$ as the factorization
gauge.  

\begin{prop}
  Let $T[y]$ be a second order rational operator that preserves a
  polynomial flag.  Let $\phi(z)$ be a quasi-rational factorization
  eigenfunction with eigenvalue $\lambda_0$.  Then, there exists a
  rational factorization \eqref{eq:TBA} such that the partner operator
  preserves a primitive polynomial flag.
\end{prop}
\begin{proof}
  Let $w(z)=\phi'(z)/\phi(z)$ and let $b(z)$ be an as yet unspecified
  rational function. Set
  \begin{align}
    \label{eq:hwdef}
    &\hw = -w-q/p+b'/b,\\
    &\hb= p/b,
  \end{align}
  and let $A[y], B[y]$ be as shown in in \eqref{eq:Adef}
  \eqref{eq:Bdef}.  An elementary calculation shows that
  \eqref{eq:TBA} holds.  Let $y_1, y_2, \ldots$ be a degree-regular
  basis of the flag preserved by $T$.  We require that the flag
  spanned by $A[y_j]$ be polynomial and primitive (no common
  factors). Observe that if we take $b(z)$ to be the reduced
  denominator of $w(z)$, then $A[y_j]$ is a polynomial for all $j$.
  However, this does not guarantee that $A[y_j]$ is free of a common
  factor.  That is indeed a stronger condition which in fact fixes the gauge
$b(z)$ up to a  choice of scalar multiple.  Finally, the intertwining relation
  \begin{equation}
    \label{eq:hTAAT}
    \hT A = A T
  \end{equation}
  implies that  $A[y_j]$ are eigenpolynomials of the partner $\hT$.
\end{proof}

In the preceding subsection, we showed that a second-order operator
$T[y]$ is formally self-adjoint relative to a weight $W$ defined by
\eqref{eq:Pdef} \eqref{eq:Wdef}.  The following Proposition describes
the effect of a factorization transformation on the corresponding
factorization function and the weight.
\begin{prop}
  \label{prop:dualbphiW}
  Suppose that rational operators 
  \[T[y]=py''+qy'+ry,\quad \hT[y]=py''+\hq y'+ \hr y\] are related by
  a rational factorization with factorization eigenfunction $\phi(z)$
  and factorization gauge $b(z)$, Then the dual factorization gauge,
  factorization eigenfunction and weight function are given by
  \begin{align}
    \label{eq:hbdef}
    & b\hb =p \\
    \label{eq:hWdef}
    &\hW/\hb =  W/b,\\
    \label{eq:hphidef}
    &\hb\hphi=1/(W\phi)
  \end{align}
\end{prop}
\begin{proof}
  Equation \eqref{eq:hbdef} follows immediately from \eqref{eq:Adef}
  \eqref{eq:Bdef} \eqref{eq:TBA}.  From there, equation
  \eqref{eq:hTAB} implies that
  \begin{equation}
    \label{eq:whwrel}
    w+\hw = -q/p+b'/b = -\hq/p + \hb'/\hb.
  \end{equation}
  Hence,
  \begin{equation}
    \label{eq:hqqrel}
    \hq = q+p'-2pb'/b.
  \end{equation}
  From here, \eqref{eq:hWdef} follows by equations \eqref{eq:Pdef}
  \eqref{eq:Wdef}.  Equation \eqref{eq:hphidef} follows from
  \eqref{eq:phidef}.
\end{proof}
The dual weights $W, \hW$ allow us to interpret the intertwining
operators $A[y], B[y]$ in terms of a formally adjoint relation
\begin{equation}
    \label{eq:ABadj}
    \int  A[f] g\,\hW dx +
    \int  B[g] f\, Wdx  =  (P/b) f g 
\end{equation}
If the right hand side vanishes on an appropriately chosen interval of
orthogonality, and if the partner operators $T, \hT$ both admit an
infinite sequence of eigenpolynomials, then the operators $T$ and $\hat T$ and
their corresponding eigenfunctions are related by a 1-step Darboux
transformation.

The dual factorization functions $\phi, \hphi$ allow us to express the
adjoint intertwiners as Wronskians:
\begin{align*}
  A[y] &=  b\phi^{-1} \cW[\phi,y]\\
  B[y] &=  \hb\hphi^{-1} \cW[\hphi,y].
\end{align*}

In Theorem \ref{thm:ExDarb} of Section \ref{sect:fact}, we established
that every $\rX_2$-operator is Darboux-connected to a classical
operator and that the requisite intertwiners also connect the
corresponding exceptional flag with the standard polynomial
flag. Theorem \ref{thm:Main} follows as an immediate corollary.   In
light of the above remarks, it is convenient  to give the connecting
intertwiners as Wronskians of factorizing functions of the classical
operators.   Therefore, before turning to the exhaustive
classification, we must review the possible quasi-rational factorizing functions for the classical operators.

\subsection{The $\rX_2$ Hermite polynomials}
The classical Hermite orthogonal polynomials are orthogonal relative
to the weight
\[ W(x) = e^{-x^2}. \] The nth Hermite polynomial
$H_n(x)$ satisfies the differential equation
\[ \cH[H_n] = - 2n H_n \]
where
\[ \cH[y] = y'' -2x y'\]

The exhaustive classification of the $\rX_2$ polynomials confirms the
factorization conjecture.  This means that all $\rX_2$ Hermite
polynomials are given as Wronskians of the classical polynomials
together with fixed quasi-rational factorization eigenfunction of the
classical Hermite operator $\cH[y]$.  These
quasi-rational eigenfunctions are listed below:
\begin{align}
  \psi^{(1)}_n(x) &=  H_n(x), &
  \cH[\psi^{(1)}] &= -2n\psi^{(1)}\\
  \psi^{(2)}_n(x) &=  e^{x^2} H_n(i x), &
  \cH[\psi^{(2)}] &= 2(n+1) \psi^{(2)}.
\end{align}

We will use $\hH_n(x)$ to denote the $\rX_2$ Hermite polynomials,
where the degree index $n$ skips exactly two values.  These
exceptional Hermite polynomials are orthogonal relative to a weight of
the form
\[ \hW(x;\alpha,\beta) = \frac{e^{-x^2}}{\xi(x)^2} \] where the
denominator $\xi(x)$ is a quadratic polynomial.  Consequently, the
$\hH_n(x)$ are eigenpolynomials of an operator of the form
\[ \hcH[y] : = \cH[y] - 2(\log \xi)' y' + r(x) y,\] where $r(x)$ is
rational in $x$ and where the prime denotes a derivative with respect
to $x$.  In order for the weight to be non-singualr, the quadratic
$\xi(x)$ must have imaginary roots.  Also, as we show below, the
rational term $r(x)=0$ always vanishes.  This is established on a
case-by-case basis, and has no apriori explanation.

\subsection{$\rX_2$-Laguerre polynomials}
The classical Laguerre weight is 
\[ W_\alpha(x) = e^{-x} x^\alpha\]
The classical Laguerre operator is 
\[ \cL_\alpha[y] := xy''+(\alpha+1-x) y' \]
The  quasi-rational eigenfunctions of this 
operator are
\begin{align}
  \phi^{(1)}_{n}(x;\alpha) &= L^{(\alpha)}_n(x) &
  \cL_\alpha[\phi^{(1)}_n] &= - n \phi^{(1)}_n \\
  \phi^{(2)}_{n}(x;\alpha) &= x^{-\alpha} L^{(-\alpha)}_n(x) &
  \cL_\alpha[\phi^{(2)}_n] &= (\alpha- n) \phi^{(2)}_n
  \\
  \phi^{(3)}_{n}(x;\alpha) &= e^x  L^{(\alpha)}_n(-x) &
  \cL_\alpha[\phi^{(3)}_n] &= (\alpha+n+1)\phi^{(3)}_n
  \\ 
  \phi^{(4)}_{n}(x;\alpha) &= e^x x^{-\alpha} L^{(-\alpha)}_n(-x) &
  \cL_\alpha[\phi^{(4)}_n] &= (n+1) \phi^{(4)}_n 
\end{align}

In confirmation of the factorization conjecture, all $\rX_2$ Laguerre
polynomials are given as first and second-order Wronskians of the
classical Laguerres and the above factorization functions.  The
$\rX_2$ polynomials themselves will be denoted by $\hL^{(\alpha)}_n$
the range of $n$ omits exactly two degrees.  In all cases, the
$\hL^{(\alpha)}_n$ are orthogonal relative to a weight of the form
\begin{equation}
    \hW(x;\alpha) := \frac{e^{-x} x^\alpha}{\xi(x;\alpha)^2},  
\end{equation}
where the denominator $\xi(x;\alpha)$ is a quadratic polynomial in
$x$.  The parameter $\alpha$ has to be restricted so that
$\xi(x;\alpha)$ has no zeros in the interval of orthogonality $x\in
(0,\infty)$.  The exceptional polynomials $\hL^{(\alpha)}_n$ arise as
eigenpolynomials of a second order operator
\[\hcL_\alpha[y] = x y'' + (1+\alpha-x) y' - 2 (\log \xi)' y' +
r(x;\alpha) y \]
where $r(x;\alpha)$ is a rational function in $x$ which will be
adjusted so that, in all cases, 
\[   \hcL_\alpha[\hL^{(\alpha)}_n] = -n \hL^{(\alpha)}_n\] 

\subsection{The $\rX_2$ Jacobi polynomials}
The classical Jacobi OPs are orthogonal relative to the weight
\[ W(x;\alpha,\beta) = (1-x)^\alpha (1+x)^\beta,\quad
\alpha,\beta>-1. \] The nth Jacobi polynomial
$P^{(\alpha,\beta)}_n(x)$ satisfies the differential equation
\[ \cT_{\alpha,\beta}[P^{(\alpha,\beta)}_n] = -
n(n+\alpha+\beta+1)P^{(\alpha,\beta)}_n \] where
\[ \cT_{\alpha,\beta}[y] = (1-x^2) y'' + (\beta-\alpha -
(2+\alpha+\beta)x) y' \]

The exhaustive classification of the $\rX_2$ polynomials confirms the
factorization conjecture.  This means that all $\rX_2$ Jacobi
polynomials are given as Wronskians of the classical polynomials
together with fixed quasi-rational factorization eigenfunction of the
classical Jacobi operator $\cT_{\alpha,\beta}[y]$.  These
quasi-rational eigenfunctions are listed below:
\begin{align}
  \phi^{(1)}_n(x;\alpha,\beta) &=  P^{(\alpha,\beta)}(x), &
  \cT[\phi^{(1)}] &= -n(n+\alpha+\beta+1) \phi^{(1)}\\
  \phi^{(2)}_n(x;\alpha,\beta) &=  (1+x)^{-\beta}P^{(\alpha,-\beta)}(x), &
  \cT[\phi^{(2)}] &= (\beta-n)(n+\alpha+1) \phi^{(2)}\\
  \phi^{(3)}_n(x;\alpha,\beta) &=  (1-x)^{-\alpha}P^{(-\alpha,\beta)}(x), &
  \cT[\phi^{(3)}] &= (\alpha-n)(n+\beta+1) \phi^{(3)}\\
    \phi^{(4)}_n(x;\alpha,\beta) &=  (1-x)^{-\alpha}(1+x)^{-\beta}
  P^{(-\alpha,-\beta)}(x), & 
  \cT[\phi^{(4)}] &= (n+1)(\alpha+\beta-n) \phi^{(4)}
\end{align}

We will use $\hP^{(\alpha,\beta)}_n(x)$ to denote the $\rX_2$
Jacobi polynomials, where the degree index $n$ skips exactly two values.
These exceptional Jacobi polynomials are orthogonal relative to a
weight of the form
\[ \hW(x;\alpha,\beta) =
\frac{(1-x)^\alpha(1+x)^\beta}{\xi(x;\alpha,\beta)^2} \] where the
denominator $\xi(x;\alpha,\beta)$ is a quadratic polynomial.  
Consequently, the $\hP^{(\alpha,\beta)}_n(x)$ are eigenpolynomials of
an operator of the form
\[ \hcT_{\alpha,\beta}[y] : = \cT_{\alpha,\beta}[y] - 2(1-x^2)(\log
\xi)' y' + r(x;\alpha,\beta) y,\] where $r(x;\alpha,\beta)$ is
rational in $x$ and where the prime denotes a derivative with respect
to $x$.  The parameters $\alpha,\beta >-1$ are so restricted in order
to have finite moments of all orders.  Additional restrictions must be
imposed on $\alpha,\beta$ to ensure that $\xi(x;\alpha,\beta)$ has no
zeros in the interval of orthogonality $x\in (-1,1)$.

\section{Classification of codimension 2 XOPs}\label{sec:XOPs}

The main result of this section is a complete list of $\rX_2$
orthogonal polynomial systems together with the intertwining operators
that connect them to the classical families of Hermite, Laguerre and
Jacobi.   The classification is summarized in the following.

\begin{thm}
\label{thm:x2op}
Up to a real affine transformation of the independent variable, all $\rX_2$
orthogonal polynomial systems are gathered in the following table:

\begin{table}[h]\label{tab:clas}\caption{Classification of $\rX_2$ orthogonal
polynomial systems}
\begin{center}
\begin{tabular}[c]{|p{1.5cm}|p{1.5cm}|p{1.5cm}|p{1.5cm}||p{1.5cm}|p{
1.5cm } |p{ 1.5cm} |} \cline{2-6}
\multicolumn{1}{r}{} & \multicolumn{1}{|l|}{ $\quad\cE^{(11)}_{23}$ } &
$\quad\cE^{(11)}_{13}$ &
$\quad\cE^{(11)}_{03}$ &  $\quad\cE^{(2a)}_{13}$&
$\quad\cE^{(2a)}_{03}$
\\ \hline
\begin{center}Hermite \end{center}&   & &
 \begin{center}\textbf{1-step}\par\S\textbf{\ref{ssect:1103Hermite}}\end{center}
& &   \\ \hline
\begin{center}Laguerre \end{center}&
\begin{center}\textbf{ 1-step }\par
\textbf{\S\ref{ssect:1123Laguerre}}\end{center}& \begin{center} 2-step \par
\S\ref{ssect:1113Laguerre}\end{center}& 
\begin{center} 1-step \par \S\ref{ssect:1103Laguerre}\end{center}& 
\begin{center} 2-step \par \S\ref{ssect:2a13Laguerre}\end{center}&
\begin{center} 2-step \par \S\ref{ssect:2a03Laguerre}\end{center} \\ \hline
\begin{center}Jacobi \end{center}&
 \begin{center}\textbf{1-step} \par\textbf{\S\ref{ssect:1123Jacobi}}\end{center}
&
 \begin{center}2-step \par\S\ref{ssect:1113Jacobi}\end{center} &
 \begin{center}1-step \par\S\ref{ssect:1103Jacobi}\end{center} & 
 \begin{center}2-step \par\S\ref{ssect:2a13Jacobi}\end{center} & 
 \begin{center}2-step \par\S\ref{ssect:2a03Jacobi}\end{center} \\
\hline
\end{tabular}
\end{center}
\end{table}
\end{thm}
In Table 2 we find the classification of $\rX_2$ orthogonal
polynomial systems. In each cell we give the number of iterated Darboux
transformations to obtain these families from a classical OPS, and we specify
the subsection where each family is described. Empty cells mean that an OPS of
that type does not exist for the given flag, and the same is true for all the
other $\rX_2$ flags not included in the table. The cells marked in bold
correspond to $\rX_2$-OPS previously known in the literature, while all other
cases are new.

In the rest of this section we will select the $\rX_2$ operators for each of
the $\rX_2$ flags in Section \ref{sect:Classif} that can be transformed into a
well defined Sturm Liouville problem of Hermite, Laguerre or Jacobi type. We
allow affine changes of variables and basically we need to transform the
leading order of the $\rX_2$ operator into $1$, $x$ or $1-x^2$ and verify that
the weight is non-singular in the corresponding interval and it has well
defined moments of all orders. This will exclude many cases and it will impose
constraints on the remaining free parameters for the cases that survive.

\subsection{$\rX_2$-Hermite OPS}

 \subsubsection{No Hermite polynomials for the 2-pole flag $\cE^{(11)}_{23}$ }
  The leading order coefficient in \eqref{eq:T1123def}, is
  \[ -\frac{1}{2} z^2 (a_0-a_1) (a_0-a_1+4)-z (a_0 a_1-a_0-a_1^2+3
  a_1)-\frac{a_1^2}{2}+a_1 \]
  We require the coefficient of $z^2$ to vanish.
  Setting $a_1= a_0$ transforms the above into
  \[ -a_0 ( 2 z + a_0/2-1) \]
  Setting $a_1 = a_0 +4$ gives
  \[ (a_0+2)(2z-a_0/2+2) \]
  In other case, it is impossible to obtain a Hermite-like operator.

\subsubsection{No Hermite polynomials for the 2-pole flag
$\cE^{(11)}_{13}$}
  The leading order coefficient in \eqref{eq:T1113def}, is
  \[ -(c_0+c_1)\frac{z^2}{2}+ c_0\left(z-\frac{1}{2}\right)\]
  It is not possible to specialize $c_0, c_1$ so that the above
  polynomial reduces  to a constant.

 \subsubsection{1-step Hermite polynomials that span the 2-pole  flag
$\cE^{(11)}_{03}$}
  \label{ssect:1103Hermite}
  Setting $\alpha_0,\alpha_1 = -1/2, q_0 = 1$ in \eqref{eq:T1103def}
  and applying the change of varibles
  \[ z = i/\sqrt{2} x + 1/2 \]
  gives a Hermite-type operator
  \[ \hcH[y] :=  y'' - 2xy - 2(\log\xi)' y' \]
  where
  \[ \xi(x) = 1+2x^2 = -\frac{1}{2} H_2(i x).\]
  The adjoint intertwiners and the exceptional polynomials are shown
  below:
  \begin{align}
    B[y] &= e^{-x^2} W [\psi^{(2)},y]\\
    A[y] &= \frac{y'}{\xi(x)} \\
    \hH_0&= 1\\
    \hH_n&=B[H_{n-3}],\quad n=3,4,5,\ldots \\
    \hcH[\hH_n] &= -2n \hH_n\\
    A[\hH_n] &= 4n H_{n-3},\quad n=0,3,4,5,\ldots
  \end{align}
  The above polynomials are related to the CPRS exactly-solvable
  potential \cite{CPRS,felsmith} and constitute the codimension-2
  instance of the modified Hermite polynomials introduced in
  \cite{dubov}.
This family was also described independently in  \cite{Dutta-Roy} for arbitrary
codimension.
\subsubsection{No Hermite polynomials for the 2-pole flag $\cE^{(11)}_{12}$}
By inspection of \eqref{eq:T1112def}, a Hermite-type operator requires
\[ \alpha_0 = \alpha_1 = \frac{1}{2} q_0 \neq 0 \]
Applying a change of variable
\[ z= a x \]
yields the weight
\[ W(x) = e^{-2a^2 x^2}{(1-4a^2 x^2)^2} \]
To have a real weight requires $a$ to be either real, or purely
imaginary.  In the first, case, the weight is singular; in the latter
case there are no singularities  but the finite-moment condition is violated.

\subsubsection{No Hermite-type polynomials for the 1-pole  flags $\cE^{(2a)},
\cE^{(2b)}$ and $\cE^{(2c)}$.}
A real valued operator and weight requires the unique pole to be
real.  However, a Hermite-type weight requires the entire real line as
the interval of orthogonality.  Therefore, even if Hermite type weights
of the form
\[ W(x) = \frac{e^{-x^2}}{(x-b)^4} \]
do exist, since $b$ is real, the resulting weight is singular.

\subsection{$\rX_2$-Laguerre  OPS}

\subsubsection{1-step Laguerre polynomials that span the 2-pole flag
$\cE^{(11)}_{23}$}
\label{ssect:1123Laguerre}
By direct inspection of \eqref{eq:T1123def} a Laguerre-type operator
requires either $a_1=a_0+4$ or $a_1=a_0$.  We consider these two
cases in turn
\begin{itemize}
\item[(I)] Imposing $a_1=a_0+4$ in \eqref{eq:T1123def}, making an affine
  change of variable 
  \[ x = (a_0+2)(z-a_0/4 -1), \] and setting
  \[ \alpha = a_0(4+a_0)/4\] gives the operator
  \[ \hcL_\alpha[y] := x y'' + (1+\alpha-x) y' - 2 (\log \xi)'(x y' + \alpha y)
\] where
  \[ \xi(x;\alpha) = L^{(\alpha-1)}_2(-x) =
  (x^2+2(\alpha+1)x+\alpha(\alpha+1))/2\] and the prime symbol denotes
  the derivative with respect to $x$.  We impose $\alpha>0$ in order
  to avoid positive zeros of $\xi(x;\alpha)$.  The resulting
  orthogonal polynomials are codimension-2 instances of the type I
  exceptional Laguerre polynomials \cite{SO1,GKM10}. The corresponding
  polynomials and the adjoint intertwining relation are shown below:
    \begin{align} 
      A[y] &:= x^{\alpha+1}    \cW[x^{-\alpha}, y]/ \xi(x;\alpha)\\
      B[y] &:= e^{-x}\cW[\phi^{(3)}_2(x;\alpha-1),    y]\\
      \hL_n^{(\alpha)}(x) &= B[L^{(\alpha-1)}_{n-2}],\qquad n=2,3,4,\ldots \\
      A[\hL^{(\alpha)}_n]&= (\alpha+n)L^{(\alpha-1)}_{n-2}.
    \end{align}
  \item[(II)] Imposing  $a_1=a_0$ in \eqref{eq:T1123def},  making an
    affine change of variable
    \[ x= a_0(4z-2+a_0)/4 ,\] and setting
    \[ \alpha = a_0^2/4-1 \]gives the operator
    \[ \hcL_\alpha[y] := xy'' +(1+\alpha-x) y' - 2x \left(\log \xi\right)'
    (y'-y) \] 
    where
    \[ \xi(x;\alpha) = L^{(-\alpha-1)}_2(x) = (x^2+ 2(\alpha-1)x +
    \alpha^2-\alpha)/2 , \quad \alpha>1\] The resulting orthogonal
    polynomials are codimension-2 instances of the type II exceptional
    Laguerre polynomials \cite{SO1,GKM10}. The definition of these
    polynomials and the adjoint differential relation are shown below
    \begin{align}
      A[y]&:= \frac{e^{-x}}{\xi(x;\alpha)} \cW[e^x,y] \\
      B[y] &:= x^{\alpha+2} \cW[\phi^{(2)}_2(x;\alpha+1),y] \\
      \hL^{(\alpha)}_n &= B[L^{(\alpha+1)}_{n-2}],\quad n=2,3,4,5,\ldots \\
      A[\hL^{(\alpha)}_n] &= (3-\alpha-n) L^{(1+\alpha)}_{n-2}
    \end{align}
  \end{itemize}

\subsubsection{2-step Laguerre polynomials that span the 2-poles flag
    $\cE^{(11)}_{13}$} 
  \label{ssect:1113Laguerre}
  By direct inspection of \eqref{eq:T1113def}, a Laguerre-type
  operator requires $c_0=1,c_1 = 0$.  Applying the affine
  transformation
  \[ x= (z-1/2) \frac{a_0 (2+a_0)}{a_0+1} \] and setting
  \[ \alpha = \frac{a_0^2 + 2 a_0 + 2}{2(a_0+1)} \]
  gives the operator
\[\hcL_\alpha[y] := xy'' +(1+\alpha-x) y' - 2x(\log \xi)' y' +
\frac{2(\al-1)(\al+1-x)}{ \xi(x;\alpha)} y\]
 where
  \[ \xi(x;\alpha) = x^2+1-\alpha^2 =
  e^{-2x} x^{1+\alpha}
  \cW[\phi^{(4)}_1(x;\alpha),\phi^{(3)}_1(x;\alpha)],\quad
  |\alpha|<1 \] The adjoint 
  intertwiners and the exceptional polynomials are:
  \begin{align}
    B_\alpha[y] &:= \frac{1}{\alpha} e^{-2x} x^{2+\alpha} \cW[
    \phi^{(3)}_1(x;\alpha),
    \phi^{(4)}_1(x;\alpha),y]\\
    \hL^{(\alpha)}_1 &:= L^{(\alpha)}_1(-x) =  x+\alpha+1\\
    \hL^{(\alpha)}_n &:= B_\alpha[ L^{(\alpha)}_{n-3}],\quad
    n=3,4,5,\ldots\\
    A_\alpha[y] &:= \frac{x^{2+\alpha}}{\alpha \xi(x;\alpha)^2}
    \cW[x^{-\alpha}(x-\alpha+1),x+\alpha+1,y] \\
    A_\alpha[\hL^{(\alpha)}_n] &= -(n-1)(\alpha+n-1)
    L^{(\alpha)}_{n-3},\quad n=1,3,4,5,\ldots
  \end{align}
  Note: for $\alpha=0$, the above definitions have to be treated as a limit
  process.  A straightforward calculation shows that
  \begin{align}
    B_0[y] &= -x(1+x^2)y''+(2x^3+x^2+2x-1) y' -(x^3+x^2+2x-2) y
  \end{align}

 \subsubsection{1-step Laguerre polynomials that span the 2-poles flag
    $\cE^{(11)}_{03}$}
  \label{ssect:1103Laguerre}
  Inspection of \eqref{eq:T1103def} reveals that a Laguerre-type
  operator requires
  \[ q_0 + c_0 + c_1 = 1\]
  Since we are free to scale the operator, no generality is lost by
  imposing $c_0-c_1 =1$, which gives us
  \[ c_0  = (1-q_0)/2,\quad c_1= -(1+q_0)/2 \]
  Applying the affine change of variables
  \[ x = q_0(1-q_0 - 2z) \]
  and setting 
  \[ \alpha= 1-q_0^2 \]
  gives the operator
  \begin{align}
    \hcL_k[y]&:=  x y'' + (1+\alpha-x) y' - 2x (\log \xi)' y' 
    \intertext{where}
    \xi(x;\alpha) &= L^{(-\alpha-1)}_2(-x) = (x^2+2(1-\alpha)x +
    \alpha^2-\alpha)/2 ,
  \end{align}
  and where 
  \[\alpha \in (-1,0) \cup (1,\infty)\] in order to avoid
  positive zeros in $\xi(x;\alpha)$ and to have finite moments.  The
  corresponding exceptional polynomials and intertwiners are shown
  below:
  \begin{align}
    B[y] &:= e^{-x} x^{2+\alpha} \cW[\phi^{(4)}_2(x;1+\alpha), y ]\\
    A[y] &= \frac{y'}{\xi(x;\alpha)}\\
    \hL^{(\alpha)}_0(x) &= 1\\
    \hL^{(\alpha)}_n(x) &= B[L^{(\alpha+1)}_{n-3}],\quad n=3,4,5,\ldots\\
    A[\hL^{(\alpha)}_n] &= n L^{(\alpha+1)}_{n-3},\quad n=0,3,4,5,\ldots
  \end{align}

\subsubsection{No Laguerre polynomials for the 2-poles flag  $\cE^{(11)}_{12}$}
By inspection of \eqref{eq:T1112def}, $q_0= c_0 + c_1$.  Without loss
of generality,  
\[ c_0 - c_1 = 1,\quad c_0 + c_1 = a \]
where $a$ is a new operator parameter.  Making the affine change of
variables
\[ x = a((1+a)-2z) \] 
gives the weight
\[ \hW_\alpha(x) = \frac{e^{-x} x^{a^2-1}}{(x-a^2-a)^2(x-a^2+a)^2} \]
In order to have a real weight we need $a$ to be either real or pure
imaginary.  In the first case, the denominator will have a positive
zero; the weight is singular.  In the former case, the finite moment
condition is violated.  Therefore, there are no $\rX_2$ polynomials
that span this flag.

\subsubsection{1-step Laguerre polynomials for the 1-pole flag
$\cE^{(2a)}_{13}$}
\label{ssect:2a13Laguerre}
We refer to the $\cE^{(2)}$ flags and the corresponding OPS as
1-pole because the weight function has one pole,
unlike the 2-poles present in the weight functions of the $\cE^{(11)}$
families. This pole in the weight has higher multiplicity.

By direct inspection of \eqref{eq:T2a13def}, a Laguerre-type operator
requires either $a=1/3$, or $a=3$.  Setting $a=1/3$,  making the
change of variables
$x = z+3/4$ yields a singular weight, namely
\[ \frac{e^{-x} x^{-1/4}}{(4x-3)^4} \]
Setting  $a=3$ and making the change of variables
\[ x=3z-3/4 \]
gives the operator
\[ \hcL[y] := x y''+(5/4-x) y' - \frac{ 4x y' + y}{x+3/4}\]
and the  weight
\[ \hW(x) = \frac{e^{-x} x^{1/4}}{(4x+3)^4},\] which is both
non-singular and has finite moments of all orders.  The remarkable
feature of this weight is that it has a  fourth order
pole, unlike the two second order poles of the previously discussed
$\rX_2$ families.  The adjoint intertwiners and the exceptional
polynomials for this weight are shown below:
\begin{align}
  B[y]&:= \frac{e^{-2x} x^{9/4}}{(x+3/4)} \cW[\phi^{(4)}_1(x;1/4)
  ,\phi^{(3)}_2(x;1/4) ,y] \\
  \hL_1(x) &:= x+15/4\\
  \hL_n(x) &:= B[L^{(1/4)}_{n-3}],\quad n=3,4,5,\ldots\\
  A[y] &:= \frac{x^{9/4}}{\left(x+3/4\right)^3}
  \cW[x^{-1/4},x+15/4,y]\\
  A[\hL_n] &= \frac{25}{128} (n-1) (4n+1) L^{(1/4)}_{n-3},\quad n=1,3,4,5,\ldots
\end{align}

\subsubsection{2-step Laguerre polynomials for the 1-pole flag
$\cE^{(2a)}_{03}$}
\label{ssect:2a03Laguerre}
By inspection of \eqref{eq:T2a03def}, a Laguerre-type operator
requires $q_0=3$.  Making the affine change of variable
\[ x = \frac{3}{4}(2z-1) \]
gives the operator
\[ \hcL[y] := xy'' + \left(3/4-x\right) y' - \frac{4 x
  y'}{x+3/4} \] 
and the weight
\[ \hW(x) := \frac{e^{-x} x^{-1/4}}{(4x+3)^4} \]
The adjoint intertwiners and the exceptional polynomials are shown below:
\begin{align}
  B[y] &:= \frac{e^{-2x} x^{7/4}}{x+3/4}
  \cW[\phi^{(4)}_2(x;-1/4), \phi^{(3)}_1(x;-1/4), y]\\
  \hL_0 &= 1 \\
  \hL_n &:= B[L^{(-1/4)}_{n-3}],\quad n=3,4,5,\ldots \\
  A[y]&:= \frac{x^{7/4}}{\left(x+3/4\right)^3}
  \cW[1,x^{1/4}(x+15/4),
  y]\\
  A[\hL_n] &= \frac{25}{128} n(5-4n) L^{(-1/4)}_{n-3},\quad
  n=0,3,4,5,\ldots
\end{align}

\subsubsection{No Laguerre polynomials for the 1-pole flags
$\cE^{(2a)}_{02},
  \cE^{(2a)}_{12},\cE^{(2b)}_{23}$ and $\cE^{(2c)}_{23}$}
Setting $p_0=0$ and applying  an affine transformation, the operator
\eqref{eq:T2a12def}  yields a singular Laguerre-type weight
\[ \hW(x) = \frac{e^{-x} x^{1/4}}{(4x-3)^4} \] By direct inspection of
\eqref{eq:T2a02def}, \eqref{eq:T2b23def} \eqref{eq:T2c23def}, the
operators in question do not admit a Laguerre form.

\subsection{$\rX_2$ -Jacobi OPS}

 \subsubsection{1-step Jacobi polynomials that span the 2-pole flag
$\cE^{(11)}_{23}$}
  \label{ssect:1123Jacobi}
  The quadratic coefficient of $y''$ in \eqref{eq:T1123def} factors as
  \begin{align}
    &  -\frac{1}{2} (a_1-a_0)(a_1-a_0-4)(z-z_1)(z-z_2) \intertext{where}
  &z_1 = \frac{a_1} {a_1-a_0-4},\quad z_2 = \frac{a_1-2}{a_1-a_0-4}    
  \end{align}
  We seek an affine change of variable that transforms this quadratic
  into $1-x^2$.  There are two possibilities according to which root
  is sent to $+1$ or $-1$.  However, since the two resulting families
  are related by an affine change of variable, it suffices to consider
  just one such transformation.  Employing the transformation
  \[ z = \frac{z_2}{2}(x+1)-\frac{z_1}{2}  (x-1) \]
   setting
  \[ \alpha = \frac{2(z_1-1)z_1(2z_2-1)}{z_1-z_2},\quad \beta =
  \frac{2(2 z_1-1) z_2 (z_2-1)}{z_1-z_2} \]
  and adding a constant term,
  transforms $T^{(11)}_{23}[y]$ into the operator
  \[ \hcT_{\alpha,\beta}[y] = \cT_{\alpha,\beta}[y] - 2  (\log \xi)'
  ((1-x^2) y'+ \beta(1-x)y) +2(\alpha-\beta-1) y\]
  where 
  \begin{align}
     \xi(x;\alpha,\beta) &= P^{(-\alpha-1,\beta-1)}_2(x) \\
     &= \frac{1}{4}  \binom{\beta-\alpha+2}{2}(x-1)^2+
  \frac{1}{2}(\beta-\alpha+1)(1-\alpha)(x-1) +\binom{\alpha}{2}
  \end{align}
  In this way, we have arrived at the codimension-2 instance of the
   exceptional Jacobi-type polynomials introduced by Odake and Sasaki
  \cite{SO1,GKM12}.

  We require that $\xi(x;\alpha,\beta)$ have no zeros in the interval
  of orthogonality $x\in [-1,1]$.  The above affine transformation
  maps $-1,1$ to the roots of $\xi(x)$ and maps
  \begin{align}
     z_1 &=  \frac{1}{2} \pm\frac{1}{2} \sqrt{-a(1+a+b)/b}
     ,\quad a=\alpha-1,\; b=-\beta-1\\
     z_2 &=\frac{1}{2} \mp\frac{1}{2}
     \sqrt{-b(1+a+b)/a} 
  \end{align}
  to $\pm 1$.   Therefore, an equivalent condition is that $z_1,z_2$
  are either complex-valued or lie in the interval $(0,1)$.  The
  solutions to this constraint in the $(a,b)$ plane are  the disjoint
  union of the following regions: (i) $a,b>0$; (ii) $a>0, b<-1$; (iii)
  $a<-1, b>0$; (iv) $-1<a,b<0$.  Finite moments require
  $\alpha,\beta>-1$.  Therefore, in the final analysis, we have two
  classes orthogonal polynomials with a non-singular weight and finite
  moments: $\alpha>-1, \beta>0$ and $0<\alpha<1, -1<\beta<0$; c.f.,
  Proposition 4.5 of \cite{GKM12}.

  The exceptional polynomials and the adjoint intertwiners are shown below:
  \begin{align}
    A[y] &:= \frac{(1+x)^{\beta+1}}{\xi(x;\alpha,\beta) }\cW[(1+x)^{-\beta},y]\\
    B[y] &:=  (1-x)^{\alpha+2}\cW[\phi^{(2)}_2(x;\alpha+1,\beta-1),y ]\\
    \hP^{(\alpha,\beta)}_n &=    B[P^{(\alpha+1,\beta-1)}_{n-2}   ]\\
    \hcT_{\alpha,\beta} &= BA + (2+\beta)(\alpha-1)\\
    \cT_{\alpha+1,\beta-1} &= AB+ (2+\beta)(\alpha-1)\\
    \hcT[\hP_n] &= -(n-2)(n-1+\alpha+\beta) \hP_n     \\
    A[\hP_n^{(\alpha,\beta)}] &= -(\alpha+n-3)(\beta+n)P^{(\alpha+1,\beta-1)}_{n-2}
  \end{align}

  \subsubsection{2-step Jacobi polynomials that span the 2-pole flag
$\cE^{(11)}_{13}$}
  \label{ssect:1113Jacobi}
  The quadratic coefficient of $y''$ in \eqref{eq:T1113def} factors as 
  \[ \frac{c_0}{2}((R+1)z -1)((R-1)z +1),\quad\text{where } R=
  \sqrt{-\frac{c_1}{c_0}}\]
  Employing the affine transformation
  \[ z = \frac{R x+1}{1-R^2} \]
  and setting
  \begin{align}
    \alpha &= \frac{1}{1-R} + \frac{a_0}{1-R} - \frac{R}{(1+a_0)(1-R)}\\
    \beta &= \frac{1}{1+R} + \frac{a_0}{1+R} + \frac{R}{(1+a_0)(1+R)}
  \end{align}
  transforms the operator $T^{(11)}_{13}$ into
  \[ \hcT_{\alpha,\beta}[y] = \cT_{\alpha,\beta}[y] -2(1-x^2)(\log
  \xi)' y' - \frac{8(\alpha-1)(\beta-1)
    P^{(\alpha,\beta)}_1(x)}{\xi(x;\alpha,\beta)}y \]
  where
  \begin{align}
    \xi(x;\alpha,\beta) &= (x^2+1)(\alpha^2-\beta^2) + 2x
    (\alpha^2+\beta^2-2) \\
    &= \frac{4 a_0 (2+a_0) (1+a_0-R)(1+a_0+R)}{(1+a_0)^2(R^2-1)^2}
    (x+R)(Rx +1)
  \end{align}
  For a real, non-singular weight, we require $R=e^{i t},\; t\in
  \Rset$ to be a unit-length complex number.  A direct calculation
  shows that 
  \begin{align*}
     R &= \frac{\alpha^2+\beta^2-2}{\alpha^2-\beta^2} \pm
  \frac{2\sqrt{(\alpha^2-1)(\beta^2-1)} }{\alpha^2-\beta^2} \\
  \frac{1}{R} &= \frac{\alpha^2+\beta^2-2}{\alpha^2-\beta^2} \mp
  \frac{2\sqrt{(\alpha^2-1)(\beta^2-1)} }{\alpha^2-\beta^2} 
  \end{align*}
Therefore, the parameters $\alpha,\beta$ must satisfy
  \[ -1<\alpha<1,\; \beta>1,\quad\text{or}\quad \alpha>1,\;
  -1<\beta<1 \]
  The corresponding exceptional polynomial and the adoint intertwiners
  are shown below
  \begin{align}
    A[y]&:= \frac{(1+x)^{\beta+2}}{\beta \xi(x;\alpha,\beta)}
    \cW[(1+x)^{-\beta}P^{(\alpha,\beta-2)}_1,P^{(-\alpha-2,\beta)}_1,y]
    \\
    B[y] &:= \frac{(1-x)^{6+2\alpha}(1+x)^{2+\beta}}{\beta}
    \cW[\phi^{(2)}_1(x;\alpha+2,\beta) ,
    \phi^{(4)}_1(x;\alpha+2,\beta),y] \\
    \hP^{(\alpha,\beta)}_1 &= P^{(-\alpha-2,\beta)}_1 \\
    \hP^{(\alpha,\beta)}_n &= B\left[P^{(\alpha+2,\beta)}_{n-3}\right],\quad
    n=3,4,5,\ldots \\
    \hcT[\hP_n] &= -n(n-3+\alpha+\beta) \hP_n\\
    A[\hP^{(\alpha,\beta)}_n] &=
    \frac{1}{16}(n-1)(n+\alpha-2)(n+\beta-1)(n+\alpha+\beta-2)
    P^{(\alpha+2,\beta)}_{n-3},\quad n=1,3,4,5,\ldots 
  \end{align}
  As above, for the case of $\beta=0$, the definitions above must be
  treated as a limit.

  \subsubsection{1-step Jacobi polynomials that span the 2-pole flag
$\cE^{(11)}_{03}$}
  \label{ssect:1103Jacobi}
  The quadratic coefficient of $y''$ in \eqref{eq:T1103def} factors as
  \begin{align}
    & -\left(q_0+c_0 +c_1\right)\frac{z^2}{2}+
      \frac{q_0z}{2}+c_0\left(z-\frac{1}{2}\right)= -\frac{1}{2
      z_1 z_2} (z-z_1)(z-z_2) \intertext{where,} &c_0 =
    1,\; c_1 =
    \left(1-\frac{1}{z_1}\right)\left(1-\frac{1}{z_2}\right),\quad q_0
    = -2+\frac{1}{z_1} + \frac{1}{z_2}
  \end{align}
  Note that no generality is lost by scaling $c_0 = 1$ because, if
  $c_0 =0$, then the operator does not have a pole at $z=0$.
  Employing the affine transformation
  \[ z = \frac{z_1(x+1)-z_2(1- x)}{2}\]
  and setting
  \[ \alpha = \frac{2(z_1-1)z_1(2z_2-1)}{z_1-z_2},\quad \beta =
  -\frac{2(2 z_1-1) z_2 (z_2-1)}{z_1-z_2} \]
  and adding a constant term,
  transforms $2 z_1 z_2 T^{(11)}_{23}$ into the operator
  \[ \hcT_{\alpha,\beta}[y] = \cT_{\alpha,\beta}[y] - 2  (\log \xi)'
  (1-x^2) y'\]
  where 
  \begin{align}
     \xi(x;\alpha,\beta) &= P^{(-\alpha-1,-\beta-1)}_2(x) \\
     &= \frac{1}{4}  \binom{2-\beta-\alpha}{2}(x-1)^2+
  \frac{1}{2}(1-\beta-\alpha)(1-\alpha)(x-1) +\binom{\alpha}{2}
  \end{align}

  We require that $\xi(x;\alpha,\beta)$ have no zeros in the interval
  of orthogonality $x\in [-1,1]$.  The above affine transformation
  maps $-1,1$ to the roots of $\xi(x)$ and maps
  \begin{align}
     z_1 &=  \frac{1}{2} \pm \frac{1}{2}\sqrt{-a(1+a+b)/b}
     ,\quad a=\alpha-1,\; b=\beta-1\\
     z_2 &=\frac{1}{2} \mp
     \frac{1}{2}\sqrt{-b(1+a+b)/a} 
  \end{align}
  to $\pm 1$.  Therefore, an equivalent condition is that $z_1,z_2$
  are either complex-valued or lie in the interval $(0,1)$.  This
  constraint, toghether with the finite moments constraint, gives us 4
  disjoint classes of acceptable parameter values:
  \begin{itemize}
  \item[(i)] $\alpha,\beta >1$;
  \item[(ii)] $1<\alpha<3,\; -1<\beta<0,\; \alpha+\beta<2$;
  \item[(iii)] $1<\beta<3,\; -1<\alpha<0,\; \alpha+\beta<2$;
  \item[(iv)] $0<\alpha,\beta<1$
  \end{itemize}

  The exceptional polynomials and the adjoint intertwiners are shown below:
  \begin{align}
    A[y] &:= \frac{y'}{P^{(-\alpha-1,-\beta-1)}_2(x)}\\
    B[y] &:=
    (1-x)^{2+\alpha}(1+x)^{2+\beta}\cW[\phi^{(4)}_2(x;\alpha+1,\beta+1),y
    ]\\ 
    \hP^{(\alpha,\beta)}_0 &= 1\\
    \hP^{(\alpha,\beta)}_n &=    B[P^{(\alpha+1,\beta+1)}_{n-3}],\quad
    n=3,4,5,\ldots\\
    \hcT[\hP_n] &= -(n-2)(n-1+\alpha+\beta) \hP_n     \\
    A[\hP_n^{(\alpha,\beta)}] &=
    -n(\alpha+n-3)P^{(\alpha+1,\beta+1)}_{n-3},\quad n=0,3,4,5,\ldots
  \end{align}

\subsubsection{No Jacobi polynomials for the 2-pole flag $\cE^{(11)}_{12}$}
The quadratic coefficient of $y''$ in \eqref{eq:T1112def} factors as
  \begin{align}
    & \left(c_0 +c_1-q_0\right)\frac{z^2}{2}+
      \left(\frac{q_0}{2}-c_1\right)-\frac{c_0}{2}= -\frac{1}{2
      z_1 z_2} (z-z_1)(z-z_2) \intertext{where,} &
    c_0 =    1,\; c_1 =
    \left(1-\frac{1}{z_1}\right)\left(1-\frac{1}{z_2}\right),\quad q_0
    = 2-\frac{1}{z_1} - \frac{1}{z_2} + \frac{2}{z_1 z_2}
  \end{align}
  Note that no generality is lost by scaling $c_0 = 1$ because, if
  $c_0 =0$, then the operator does not have a pole at $z=0$.
  Employing the affine transformation
  \[ z = \frac{z_1(x+1)-z_2(1- x)}{2}\]
  and setting
  \[ \alpha = -\frac{2(z_1-1)z_1(2z_2-1)}{z_1-z_2},\quad \beta =
  \frac{2(2 z_1-1) z_2 (z_2-1)}{z_1-z_2} \]
  gives a weight of the form
  \[ \hW(x;\alpha,\beta ) = \frac{(1-x)^\alpha
    (1+x)^\beta}{\left(P^{(\alpha-1,\beta-1)}_2(x)\right)^2} \] 
   Since
  \begin{align}
     z_1 &=  \frac{1}{2} \pm \frac{1}{2}\sqrt{a(1+a+b)/b}
     ,\quad a=\alpha+1,\; b=\beta+1\\
     z_2 &=\frac{1}{2} \mp
     \frac{1}{2}\sqrt{b(1+a+b)/a} 
  \end{align}
  and since $\alpha,\beta>-1$ is required for finite moments, the
  roots $z_1,z_2$ are real, and one of them lies outside the interval
  $(0,1)$. Therefore, if $\alpha,\beta>-1$, the above weight must be
  singular on $x\in (-1,1)$.

\subsubsection{2-step Jacobi polynomials that span the 1-pole flag
$\cE^{(2a)}_{13}$}
\label{ssect:2a13Jacobi}
The quadratic coefficient of $y''$ in \eqref{eq:T2a13def} factors as
\[\left( \left(1-3 a\right)
        \left(3-a\right)\frac{z^2}{4} +2 \left(1-a\right)z+1\right) 
= \frac{1}{4}((a-3)z -2)((3a-1)z-2) \]
In order to have a Jacobi-type operator, we require $a\neq 3,1/3,-1$;
in the latter case we obtain a perfect square.  Applying the affine
transformation 
\[ z = \frac{(x+1)}{a-3} -\frac{x-1}{3a-1} \]
yields the operator
\[ \hcT_a[y]:= \cT_{\alpha,\beta}[y] - 4(1-x^2)(\log\xi)' y'-
\frac{8}{\xi(x;a)}  y \]
where
\[ \xi(x;a) = (1+a) x + 2(a-1) \]
and where
\[ \alpha = 2+ \frac{6}{a-3},\quad \beta = \frac{2}{3a-1} \]
Just as for the Laguerre-type polynomials, the corresponding weight
involves a 4th order pole: 
\[ \hW(x;a) = \frac{(1-x)^\alpha(1+x)^\beta}{\xi(x;a)^4} \] In order
to obtain a non-singular weight we must have $a>3$ or $a<1/3$.
However, in order to have $\alpha,\beta>-1$ (finite moments), we must
restrict the latter condition to $a<-1/3,\; a\neq -1$.  The
corresponding values of $\alpha,\beta$ range from $\alpha>2,\;
0<\beta<2$ in the former case, and $1/5<\alpha<2,\; -1<\beta<0,\;
(\alpha,\beta)\neq (1/2,-1/2)$ in the latter case.  Of course $\alpha,
\beta$ are not independent, but rather are linked by the relation
\[ 4\alpha\beta+\beta-\alpha+2 = 0 \]
The adjoint intertwiners and the exceptional polynomials for this
flag and weight are shown below:
\begin{align}
  B[y] &:=
  \frac{(1-x)^{2\alpha+6}(1+x)^{\beta+2}}{a(a-1)(1+3a)\xi(x;a)}
  \cW[\phi^{(4)}_1(x;\alpha+2,\beta) , \phi^{(2)}_2(x;\alpha+2,\beta),
  y] \\
  A[y] &:= \frac{(3a-1)^5 (a-3)^3}{36(1+3a) \xi(x;a)^3}
  \cW[(1+x)^{-\beta}, 2(1+a)(x-1)+(a-1)(3a-1),y]\\
  \hP_1(x;a) &=  2(1+a)(x-1)+(a-1)(3a-1)\\
  \hP_n(x;a) &:= B[P^{(\alpha,\beta)}_{n-3}],\quad
  n=3,4,5,\ldots,\\
  A[\hP_n1] &=
  (n-1)(n-3+\alpha)(n+\beta)(n-2+\alpha+\beta)P^{(\alpha,\beta)}_{n-3},\quad n=1,3,4,5,\ldots
\end{align}

\subsubsection{2-step Jacobi polynomials that span the 1-pole flag
$\cE^{(2a)}_{03}$}
\label{ssect:2a03Jacobi}
The quadratic coefficient of $y''$ in \eqref{eq:T2a03def} factors as
\[\left( (3-q_0)\frac{z^2}{4} -2z+1\right)  =
\frac{(z-z_1)(z-z_2)}{z_1^2}\]
where
\[ z_1, z_2 = \frac{-4 \pm 2 \sqrt{1+q_0}}{q_0-3} ,\quad z_2 =
\frac{z_1}{2z_1-1}.\] Applying the affine transformation
\[ z = \frac{z_1(x+1)}{2} -\frac{(x-1)z_2}{2},\quad z_1\neq z_2 \]
yields the operator
\[ \hcT[y]:= \cT_{\alpha,\beta}[y] - 4(1-x^2)(\log\xi)' y'\]
where
\[ \xi(x;z_1) = (z_1-1) x + z_1 \]
and where
\[ \alpha = \frac{3}{2} z_1 -1,\quad \beta = \frac{3}{2} z_2 -1,\quad 4\alpha\beta+\alpha+\beta-2=0\]
Just as for the Laguerre-type polynomials, the corresponding weight
involves a 4th order pole: 
\[ \hW(x;z_1) = \frac{(1-x)^\alpha(1+x)^\beta}{\xi(x;z_1)^4} \] In
order to obtain a non-singular weight we require $z_1\neq z_2$ to have
the same sign.  This implies that $z_1>1/2,\; z_1\neq 1$, which in
turn implies that $\alpha,\beta>-1/4,\; \alpha,\beta\neq 1/2$
but are subject to the relation
\[ 4\alpha\beta+\alpha+\beta-2=0\]
The finite moment condition is therefore automatically satisfied.
The adjoint intertwiners and the exceptional polynomials for this
flag and weight are shown below:
\begin{align}
  B[y] &:=
  \frac{(1-x)^{2\alpha+6}(1+x)^{\beta+2}}{P^{(-\alpha-2,\beta)}_1(x)}
  \cW[\phi^{(4)}_2(x;\alpha+2,\beta) , \phi^{(2)}_1(x;\alpha+2,\beta),
  y] \\
  A[y] &:= \frac{2(1+\alpha)^3  (1+x)^{2+\beta}}{(\beta-1)^2 \alpha (\alpha-2)^2}
  \cW[1,(1+x)^{-\beta}(1+\alpha+(x-1)\beta(1-2\alpha),y]\\
  \hP_0(x;z_1) &=  1\\
  \hP_n(x;z_1) &:= B[P^{(\alpha+2,\beta)}_{n-3}],\quad
  n=3,4,5,\ldots,\\
  A[\hP_n] &=
  n(n-2+\alpha)(n-1+\beta)(n-3+\alpha+\beta)P^{(2+\alpha,\beta)}_{n-3},\quad
  n=0,3,4,5,\ldots
\end{align}

\subsubsection{No Jacobi polynomials for the 1-pole flags $\cE^{(2a)}_{02},
  \cE^{(2a)}_{12},\cE^{(2b)}_{23}, \cE^{(2c)}_{23}$}

Setting
\[ z_1,z_2 =  \frac{-1\pm \sqrt{1-p_0}}{p_0} \] and applying the
affine change of variables
\[ z = \frac{z_1(x+1)}{2} -\frac{(x-1)z_2}{2},\quad z_1\neq z_2 \]
transforms the operator in \eqref{eq:T2a12def} into Jacobi form.
The corresponding weight is
\[ \hW(x;z_1) = \frac{(1-x)^\alpha (1+x)^\beta}{(x(z_1+1)+z_1)^4}\]
where 
\[ \alpha = -1+\frac{3}{2} z_1,\quad \beta = -1+\frac{3}{2} z_2 \]
A non-singular weight requires that $z_1, z_2$ be real and have the
same sign.  Since
\[ z_2 = \frac{-z_1}{2z_1+1} \] the only possibility is that
$z_1,z_2<-1/2$.  However, this means that $\alpha,\beta<-1$, which
violates the finite moments condition.

By direct inspection of \eqref{eq:T2a02def} \eqref{eq:T2b23def}, a
Jacobi-type operator must have a singularity at $x=0$.  The
coefficient of $y''$ in \eqref{eq:T2c23def} is a perfect square, which
does not permit a Jacobi-type operator.

\section{Summary and outlook}

In the present paper we have given a classification of exceptional orthogonal polynomial systems of codimension two ($\rX_2$-OPS).
The classification includes all the cases previously known in codimension two plus some new examples of exceptional polynomials.
Among the new families, the one-pole flags are clearly special. Generically, the
weight of a $\rX_m$-OPS is a rational modification of a classical weight with
$m$ double poles, and this is the case for all the families known to date. The
Jacobi and Laguerre OPS that span the $\cE^{(2a)}$ flag have codimension two but
only one pole in their weight, with quadruple multiplicity. They also have one
less free parameter than the usual Laguerre and Jacobi families, i.e. no free
parameters for the $\cE^{(2a)}$-Laguerre and just one free parameter for the
$\cE^{(2a)}$-Jacobi. The explanation for the presence of these exotic families
is that generically they would belong to a higher-codimensional family, but that
a careful tuning of the parameters can make the codimension drop by one and have
two of the poles of the weight coalesce. Thus, the generic weight of an 
$\rX_m$-OPS is a classical weight divided by the square of a certain degree $m$
polynomial $\xi(x)$ with simple roots that lie outside the interval of
orthogonality, but we know that degenerate cases are also possible.

We have also shown that every $\rX_2$-OPS can be obtained from a classical OPS by a sequence of at most two Darboux transformations, and we conjecture this result to be true {\it mutatis mutandis} for any codimension $m$. Even if the conjecture could be proved to be true, the scheme of multiple step Darboux transformations is still very rich: there are four quasi-rational factorizing functions for the Laguerre and Jacobi families and two for the Hermite. The SL-OPS obtained by 1-step Darboux transformations have been studied in all cases, but multi-step Darboux transformations might mix factorizing functions of different kinds and all the possibilities have not yet been explored. It could also  happen that even if the intermediate weights in a multi-step Darboux transformation are singular, the final weight will be regular.
All cases when this happens have been studied for multi-step state-deleting Darboux transfomations in a more general Sturm-Liouville context (not necessarily polynomial) by Krein and Adler \cite{krein,adler}. A generalization of Krein-Adler's Theorem to multi-step isospectral transformations has been performed by Grandati \cite{Grandati1}, but the full characterization of SL-OPS obtainable via multi-step Darboux transformations of mixed type remains an open problem.

Another consequence of the conjecture is that all exceptional polynomials could be written as Wronskian determinants involving essentially classical orthogonal polynomials (more specifically, involving one classical polynomial and many quasi-rational factorizing functions). 
\vskip 0.6cm
\paragraph{\textbf{Acknowledgements}}
\thanks{

The research of DGU was supported in part by MICINN-FEDER grant MTM2009-06973
and CUR-DIUE grant 2009SGR859. The research of NK was supported in part by NSERC
grant RGPIN 105490-2011. The research of RM was supported in part by NSERC grant
RGPIN-228057-2009.
}

\bibliographystyle{plain}

\end{document}